\documentclass[11pt]{umj_eng3}
\pdfoutput=1
\usepackage[utf8]{inputenc}
\usepackage{amsmath}
\usepackage{amssymb}
\usepackage{amsfonts, dsfont, graphicx, srcltx}

\firstpage{121}

\renewcommand{\leq}{\leqslant}

\renewcommand{\leq}{\leqslant}

\usepackage{eepic,color}
\usepackage{tikz}
\usepackage{bm}
\newtheorem{theorem}{Theorem }
\newtheorem{remark}{{Remark }}
\newtheorem{lemma}{Lemma }

\numberwithin{equation}{section}
\def\ln{\mbox{ln}\;}
\def\lc{\mbox{l.c.\,}}
\def\diag{\mbox{diag\,}}
\def\im{\text{Im\,}}

\def\openone{\leavevmode\hbox{\small1\kern-3.3pt\normalsize1}}
\def\ad{\text{ad\,}}

\def\openone{\leavevmode\hbox{\small1\kern-3.3pt\normalsize1}}
\def\im{\mbox{Im\,}}

\def\diag{\mbox{diag\,}}

\def\const{\mbox{const\,}}
\def\ad{\mbox{ad\,}}

\def\openone{\leavevmode\hbox{\small1\kern-3.3pt\normalsize1}}
\def\im{\mbox{Im\,}}

\def\diag{\mbox{diag\,}}

\def\const{\mbox{const\,}}
\def\ad{\mbox{ad\,}}

\title[ On the mKdV equations related to the Kac-Moody  algebras]{ On the mKdV equations related to the Kac-Moody  algebras $A_5^{(1)}$ and $A_5^{(2)}$
}
\author{V. S. Gerdjikov}

\address{ Vladimir Stefanov Gerdjikov,
\newline\hphantom{iii} Institute of Mathematics and Informatics
\newline\hphantom{iii} Bulgarian Academy of Sciences
\newline\hphantom{iii} Acad. Georgi Bonchev Str., Block 8,
\newline\hphantom{iii} 1113 Sofia, Bulgaria
\smallskip
\newline\hphantom{iii} Sankt-Petersburg    State    University    of   Aerospace
Instrumentation
\newline\hphantom{iii} St-Petersburg,   B.Morskaya,   67A
\newline\hphantom{iii} St-Petersburg, 190000, Russia
\smallskip
\newline\hphantom{iii}  Institute for Advanced Physical Studies,
\newline\hphantom{iii} 111 Tsarigradsko chaussee,
\newline\hphantom{iii} Sofia 1784, Bulgaria
\smallskip
\newline\hphantom{iii} Institute for Nuclear Research and Nuclear Energy
\newline\hphantom{iii} Bulgarian Academy of Sciences,
\newline\hphantom{iii} 72 Tsarigradsko Chaussee, Blvd.,
\newline\hphantom{iii} 1784 Sofia, Bulgaria
}
\email{vgerdjikov@math.bas.bg}

\thanks{\sc V. S. Gerdjikov, On the mKdV equations related to the Kac-Moody  algebras $A_5^{(1)}$ and $A_5^{(2)}$}
\thanks{\copyright \ V. S. Gerdjikov, 2021}
\thanks{\rm The reported study by  V. S. Gerdjikov was funded  in part by the Bulgarian National Science Foundation  under contract  KP-06N42-2.}
\thanks{\it Submitted April 12, 2021.}

\date{ }

\begin{document}

\maketitle

{\small
\begin{quote}
\noindent\noindent{\bf Abstract.}  We outline the derivation of the mKdV equations related to the Kac-Moody algebras $A_5^{(1)}$ and $A_5^{(2)}$. First we formulate their Lax representations and provide details how they can be obtained from generic Lax operators related to the algebra $sl(6)$ by applying proper Mikhailov type reduction groups $\mathbb{Z}_h$. Here $h$ is the Coxeter number of the relevant Kac-Moody algebra. Next we adapt Shabat's method for constructing the fundamental analytic solutions of the Lax operators $L$. Thus we are able to reduce the direct and inverse spectral problems for $L$ to  Riemann-Hilbert  problems (RHP) on the union of $2h$ rays $l_\nu$. They start from the origin of the complex $\lambda$-plane and close equal angles $\pi/h$. To each $l_\nu$ we associate a subalgebra $\mathfrak{g}_\nu$ which is a direct sum of $sl(2)$-subalgebras. Thus to each regular solution of the RHP we can associate scattering data of $L$ consisting of  scattering matrices $T_\nu \in \mathcal{G}_\nu$ and their Gauss decompositions. The main result of the paper is to extract from $T_0$ and $T_1$ related to the rays $l_0$ and $l_1$ the minimal sets of scattering data $\mathcal{T}_k$, $k=1, 2$. We prove that each of the minimal sets $\mathcal{T}_1$  and $\mathcal{T}_2$  allows one to reconstruct both the scattering matrices $T_\nu$, $\nu =0, 1, \dots 2h$ and the corresponding potentials of the Lax operators $L$.

\medskip



\end{quote}
}

\section{ Introduction }

This paper is an extension of a series of papers on Kax-Moody algebras and mKdV equations \cite{GMSV4, GMSV3, GMSV1, GMSV2, GMSV5} and two other more recent papers \cite{GSIB, GMSV}.
There  we derived explicitly the system of mKdV equations related to several particular choices of Kac-Moody algebras, including some twisted ones like $D_4^{(s)}$, $s=1, 2, 3$, $A_4^{(1)}$ and $A_5^{(2)}$.

The next natural steps to be considered are to solve the direct and inverse scattering method for the relevant Lax operators and to construct their reflectionless potentials and, as a consequence the soliton solutions of the mKdV systems. The methods for doing this have been already analyzed in \cite{FaTa, GeYaV, Ger, SIAM*14, GYa, VG-Ya-13, Y}. Therefore it will not be difficult to specify the construction of the fundamental analytic solutions (FAS) \cite{Sha*75, Sha} of the relevant Lax operators  and formulate the corresponding Riemann-Hilbert problem (RHP). In constructing  the soliton solutions the most effective method known to us is the dressing Zakharov-Shabat method \cite{ZaSha1, ZaSha2}.

In Section 2 below we will outline preliminary known results about the structure of the Lax operators for the case of $A_5^{(1)}$ and $A_5^{(2)}$ Kac-Moody algebras and for the recursion operators, see \cite{GMSV}. Section 3 is devoted to the fundamental analytic solutions (FAS) and to the Riemann-Hilbert problems for both cases. In Section 4 we formulate the minimal sets of scattering data and show that they reconstruct both the potential and the sewing functions of the RHP. The appendices we explain some algebraic details about the structures of Kac-Moody algebras.

\section{Preliminaries}
\subsection{Lax representations: $A_5^{(1)}$ case}

We will assume that the readers are familiar with the theory of simple Lie algebras and Kac-Moody algebras, see \cite{Bourb, Helg, Cart} and their applications in the studies of integrable NLEE \cite{DS, DS0}. Details about the bases and the gradings of the Kac-Moody algebras are given in the appendices. We will consider here a NLEE with a simplest nontrivial dispersion law, which is $f_{\rm mKdV}(\lambda) = \lambda^3 K$.

In this section following our previous papers  we formulate the Lax pairs whose potentials are  elements of the $A^{(1)}_{5}$ and $A^{(2)}_{5}$ algebras for the mKdV equations. .
They represent the third nontrivial member in the hierarchy of soliton equations related to these  algebras. The results presented here are derived in \cite{GMSV1, GMSV4} for the case of $A^{(1)}_{5}$ and in \cite{GMSV, GSIB} for $A^{(2)}_{5}$.

We will consider the Lax pair that is polynomial in the spectral parameter $\lambda$
\begin{equation}\label{eq:A51}\begin{aligned}
L\psi\equiv &\biggl( i\frac{\partial}{\partial x}+Q(x,t)-\lambda J \biggr) \psi=0\,,\\
M\psi\equiv & \biggl( i\frac{\partial}{\partial t}+V_{0}(x,t)+\lambda V_{1}(x,t)+\lambda^{2}V_{2}(x,t)-\lambda^{3} K \biggr) \psi=-\lambda^{3} \psi K.
\end{aligned}\end{equation}

The zero-curvature condition $[L\,,M] = 0$ leads to a polynomial of fourth order in $\lambda$ which has to be identically zero. The Kac-Moody algebra $A_5^{(1)}$ is graded by the Coxeter automorphism $C_1$ (see Appendix A below). The basis we will use is given by: \begin{equation}\label{eq:Jsk}\begin{aligned}
J_s^{(k)} =& \sum_{j=1}^{6} \epsilon_{j,j+s} \omega_1^{-k(j-1)} E_{j,j+s}, &\quad  \epsilon_{j,j+s} =& \begin{cases} 1 & \mbox{if} \quad j+s \leq 6, \\ -1 & \mbox {if} \quad j+s> 6. \end{cases} \\
\left[ J_s^{(k)} , J_l^{(m)} \right] =& \left( \omega_1^{-ms} - \omega_1^{-kl} \right) J_{s+l}^{(k+m)}, &\;
J_{s}^{(k)}J_{p}^{(m)} = &\omega_1^{-sm}J_{s+p}^{(k+m)}, \\
(J_{s}^{(k)})^{-1} =& (J_{s}^{(k)})^{\dagger}.
\end{aligned}\end{equation}
The potential coefficients of the Lax pair are defined by:
\begin{equation}\label{eq:Q-V}\begin{aligned}
Q(x,t) =& \sum_{j=1}^{5}q_{j}(x,t)J_{j}^{(0)}, & \quad V_{1}(x,t) = &\sum_{l=1}^{6}v_{l}^{(1)}(x,t)J_{l}^{(1)}, &\quad
J =& J_{0}^{(1)}, \\
 V_{2}(x,t)= &\sum_{l=1}^{6} v_{l}^{(2)}(x,t)J_{l}^{(2)}, &\quad V_{0}(x,t) =& \sum_{l=1}^{5}v_{l}^{(0)}(x,t)J_{l}^{(0)} , & \quad  K=& J_{0}^{(3)}.
\end{aligned}\end{equation}

The condition $[L,M]=0$ leads to a set of recurrent relations (see \cite{GeYaV, GYa, VG-13}) which allow us to determine
$V^{(k)}(x,t)$ in terms of the potential $Q(x,t)$ and its $x$-derivatives.

Using the choices for $Q$, $J$ and $K$ from (\ref{eq:Q-V}) gives:
\begin{equation}\begin{aligned}\label{eq:Q-Jd}
Q =&  \left(\begin{array}{cccccc} 0 & q_1 & q_2 & q_3 & q_4 & q_5 \\ -q_5 & 0 & q_1 & q_2 & q_3 & q_4 \\ -q_4 & -q_5 & 0 & q_1 & q_2 & q_3 \\ -q_3 & -q_4 & -q_5 & 0 & q_1 & q_2 \\ -q_2 & -q_3 & -q_4 & -q_5 & 0 & q_1 \\ -q_1 & -q_2 & -q_3 & -q_4 & -q_5 & 0  \end{array}\right), \\
J =&  \diag (1, \omega^5, \omega^4, \omega^3, \omega^2, \omega), \qquad
K= \diag(1, -1, 1, -1, 1, -1),
\end{aligned}\end{equation}
where $\omega = \exp(2\pi i/6)$.
These equations admit the following Hamiltonian formulation
\begin{equation*}
\frac{\partial q_{i}}{\partial t}=\frac{\partial}{\partial x} \biggl( \frac{\delta H}{\delta q_{6-i}} \biggr) .
\end{equation*}

The Hamiltonian density is:
\begin{equation}\label{eq:Ha}\begin{aligned}
H =&  -32 \frac{\partial q_1}{ \partial x } \frac{\partial q_5}{ \partial x } + 2 \left( \frac{\partial q_3}{ \partial x } \right)^2 \\
 + & 8 \sqrt{3} \left(  - 2q_2q_3 \frac{\partial q_1}{ \partial x } + 2q_5^2 \frac{\partial q_2}{ \partial x } +(q_1q_2 +q_4q_5) \frac{\partial q_3}{ \partial x } + 2 q_1^2 \frac{\partial q_4}{ \partial x }  - 2 q_3q_4 \frac{\partial q_5}{ \partial x }  \right) \\
+ & 2q_3^4 -24 (q_1q_5 +q_2q_4) q_3^2 +16 (q_1^3 -3q_1q_4^2 -3q_2^2 q_5 + q_5^3)q_3 + 24 (q_1q_2 -q_4q_5)^2.
\end{aligned}\end{equation}


\subsection{ Lax representations: $A_5^{(2)}$ case}
Here we formulate the main results of a recent paper \cite{GMSV}, see also \cite{VSG-88, GIS-2019, GMSV4, GMSV3, GMSV5}. The grading  used here is described in Appendix B. It uses the Coxeter automorphism $C_2$ and splits $A_5$ into 10 subspaces.
The dispersion laws of the NLEE related to $A_5^{(2)}$ are odd functions of $\lambda$; therefore NLS-type equations here are not allowed. So we are left with $f_{\rm mKdV}(\lambda) = \lambda^{3} K$.

The Lax pair has the form
\begin{equation}\label{eq:A52}
\begin{aligned}
L =&  i \partial_x + Q(x,t) - \lambda J, \\
M =&  i \partial_t + V^{(0)}(x,t) + \lambda V^{(1)}(x,t) + \lambda^2 V^{(2)}(x,t) - \lambda^{3}K,
\end{aligned}
\end{equation}
where
\begin{equation}\label{eq:A52q}
Q(x,t) \in \mathfrak{g}^{(0)}, \quad V^{(k)}(x,t) \in \mathfrak{g}^{(k)}, \quad K \in \mathfrak{g}^{(3)}, \quad J\in \mathfrak{g}^{(1)}.
\end{equation}

Here we fix up $J$ and $K$ as follows:
\begin{equation}\label{JK}
J= \diag ( \omega_2^4, \omega_2^2, 1, 0, \omega_2^6, \omega_2^8 ), \qquad K= 20 J^3,
\end{equation}
where $\omega_2=\exp(2i\pi/10)$ and choose
\begin{equation}\label{Q}
Q= \sum_{j=1}^{3} q_j \mathcal{E}_j^{(0)} = \begin{pmatrix}
0 & q_1 & q_3 & q_2 & -q_1 & -q_3 \\
-q_1 & 0 & q_1 & -q_2 & -q_3 & -q_3 \\
-q_3 & -q_1 & 0 & q_2 & -q_3 & -q_1 \\
-q_2 & q_2 & -q_2 & 0 & q_2 & -q_2 \\
q_1 & q_3 & q_3 & -q_2 & 0 & -q_1 \\
q_3 & q_3 & q_1 & q_2 & q_1 & 0 \\
\end{pmatrix}.
\end{equation}
Next we solve the recurrent relations with the result:
\begin{equation}\label{eq:Vp}\begin{aligned}
V_p^{\rm f} =&  \sum_{j=1}^{3}v_{p;j}  \mathcal{E}_j^{(p)} , \qquad p =2,1, 0; \qquad
V_1 = V_1^{\rm f}  + v_{1;4}  J,
\end{aligned}\end{equation}
and obtain explicit expressions for $v_{p;j}$ in terms of $q_j$ and their $x$-derivatives; for details see \cite{GMSV4, GMSV}. The equations of motion:
\begin{equation}\label{eq:eq0}\begin{aligned}
 \frac{\partial q_j}{ \partial x } = \frac{\partial v_{0;j}}{ \partial x }, \qquad j=1,2,3.
\end{aligned}\end{equation}
can be cast in Hamiltonian form as follows:
\begin{equation}\label{eq:eqH}\begin{aligned}
 \frac{\partial q_j}{ \partial x } = \frac{\partial }{ \partial x } \frac{\delta H}{ \delta q_j(x) } =  \frac{\partial v_{0;j}}{ \partial x }, \qquad j=1,2,3.
\end{aligned}\end{equation}
where
\begin{equation}\begin{aligned}\label{eq:H2}
H =&  2 \left\{  ( 3 \sqrt{5} + 5) \left(  \frac{\partial q_1}{ \partial x } \right)^2
- 10 \left( \frac{\partial q_2}{ \partial x } \right)^2 - ( 3 \sqrt{5} - 5) \left( \frac{\partial q_3}{ \partial x } \right)^2  \right\} \\
+ & 20 \left( c_2^- q_3^2 + c_2^+ q_1 q_3 - 2  c_2^+ q_2^2 \right) \frac{\partial q_1}{ \partial x } - 20\left( c_2^-  q_1q_3 +  c_2^+ q_1^2  - 2 c_2^- q_2^2 \right) \frac{\partial q_3}{ \partial x } \\
 + & 40 \left( - c_2^- q_3 +  c_2^+ q_2 \right) q_2 \frac{\partial q_1}{ \partial x }
+ 20 q_2^4 + 40 q_1 q_3 (q_3^2 -q_1^2) + 60 (q_1^2 q_3 ^2 - q_2^2 q_3 ^2 - q_1^2 q_2 ^2).
\end{aligned}\end{equation}
where
\begin{equation}\label{eq:cpm}\begin{aligned}
 c_2^+ = \sqrt{2 + \frac{2}{\sqrt{5}} }, \qquad c_2^- =  \sqrt{2 - \frac{2}{\sqrt{5}} }.
\end{aligned}\end{equation}

Let us now repeat the calculations using the second type of grading, see eq. (\ref{eq:gr2c}).
In it the potential takes diagonal form while $J$ becomes the sum of admissible roots.
We can do the grading using an alternative choice of the Coxeter automorphism given by (\ref{eq:C2t}), (\ref{eq:Etij}). This means
\begin{equation}\label{eq:LM0}
\begin{aligned}
\tilde{ Q}(x,t) =&  i\sum_{j=1}^{3}u_{j}(x,t)\mathcal{E}_{jj}^{+}, \qquad
\tilde{ J} =\mathcal{E}_{21}^{+}+\mathcal{E}_{32}^{+}+\frac{1}{2} \mathcal{E}_{43}^{+} +\frac{1}{2}\mathcal{E}_{15}^{-}, \\
V^{(0)}(x,t) =& \sum_{j=1}^{3}v_{j}^{(0)}\mathcal{E}_{jj}^{+},\\
V^{(1)}(x,t) =& v_{1}^{(1)}\mathcal{E}_{21}^{+}+v_{2}^{(1)}\mathcal{E}_{32}^{+} +\frac{1}{2} v_{3}^{(1)}\mathcal{E}_{43}^{+}+\frac{1}{2}v_{4}^{(1)}\mathcal{E}_{15}^{-}, \\
V^{(2)}(x,t) =& -v_{1}^{(2)}\mathcal{E}_{31}^{+}-v_{2}^{(2)}\mathcal{E}_{42}^{+}-\frac{1}{2}v_{3}^{(2)}\mathcal{E}_{14}^{-}, \qquad \qquad \tilde{K}=5 \tilde{J}^{3},
 \end{aligned}
 \end{equation}
where
\begin{equation}\label{eq:V2}
\begin{aligned}
v_{1}^{(2)} = -5i (u_{1}+u_{2}+u_{3}),  \qquad  v_{2}^{(2)} =&  -5i u_{2}, \qquad
v_{3}^{(2)} = -5i (u_{1}-u_{2}-u_{3}),
\end{aligned}
\end{equation}
\begin{equation}\label{eq:V1}
\begin{aligned}
v_{1}^{(1)} =& 10 \left( u_{1}u_{2}-\frac {\partial u_{1}}{\partial x} \right)+v_{4}^{(1)}, \\
v_{2}^{(1)} =& 5 \left ( u_{3}^{2}-u_{1}^{2}+u_{1}u_{2}+u_{2}u_{3}+\frac {\partial }{\partial x} \left( u_{3}+u_{2}-u_{1} \right) \right)+v_{4}^{(1)}, \\
v_{3}^{(1)} =& 5 \left ( u_{3}^{2}-u_{2}^{2}-u_{1}^{2}+u_{1}u_{2}+\frac {\partial }{\partial x} \left( u_{3}+2u_{2}-u_{1} \right) \right)+v_{4}^{(1)}, \\
v_{4}^{(1)}=& 2u_{2}^{2}+2u_{1}^{2}-3u_{3}^{2}-5u_{1}u_{2} +\frac {\partial }{\partial x} (5u_{1}-4u_{2}-3u_{3})  .
\end{aligned}
\end{equation}

For $V^{(0)}(x,t)$ we find:
\begin{equation}\label{eq:V0}
\begin{aligned}
v_{1}^{(0)} =& i \biggl( -5\frac {\partial^{2}u_{1}}{\partial {x}^{2}}+3u_{1}\frac{\partial}{\partial x}(3u_{2}
+u_{3})-2u_{1}^{3}+3u_{1}(u_{2}^{2}+u_{3}^{2}) \biggr) ,\\
v_{2}^{(0)} =&  i \biggl(  \frac {\partial^{2}}{\partial {x}^{2}}(4u_{2}+3u_{3})+3u_{2}\frac{\partial u_{3}}{\partial x}
-9u_{1}\frac{\partial u_{1}}{\partial x}+6u_{3}\frac{\partial u_{3}}{\partial x}-2u_{2}^{3}+3u_{2}(u_{1}^{2}+u_{3}^{2}) \biggr) ,\\
v_{3}^{(0)} =&  i \biggl( \frac {\partial^{2}}{\partial {x}^{2}}(u_{3}+3u_{2})-6u_{3}\frac{\partial u_{2}}{\partial x}
-3u_{1}\frac{\partial u_{1}}{\partial x}-3u_{2}\frac{\partial u_{2}}{\partial x}-2u_{3}^{3}+3u_{3}(u_{1}^{2}+u_{2}^{2}) \biggr).
\end{aligned}
\end{equation}

Finally the set of mKdV equations takes the form:
\begin{equation}\label{eq:mkdv2}
\begin{aligned}
\frac{\partial u_{1}}{\partial t} =& \frac{\partial}{\partial x} \biggl( -5\frac {\partial^{2}u_{1}}{\partial {x}^{2}}
+3u_{1}\frac{\partial}{\partial x}(3u_{2}+u_{3})-2u_{1}^{3}+3u_{1}(u_{2}^{2}+u_{3}^{2}) \biggr) ,\\
\frac{\partial u_{2}}{\partial t} = & \frac{\partial}{\partial x} \biggl( \frac {\partial^{2}}{\partial {x}^{2}}(4u_{2}
+3u_{3})+3u_{2}\frac{\partial u_{3}}{\partial x}-9u_{1}\frac{\partial u_{1}}{\partial x} \\
 &\qquad \qquad +6u_{3}\frac{\partial u_{3}}{\partial x}-2u_{2}^{3}+3u_{2}(u_{1}^{2}+u_{3}^{2}) \biggr) ,\\
\frac{\partial u_{3}}{\partial t} =& \frac{\partial}{\partial x} \biggl( \frac {\partial^{2}}{\partial {x}^{2}}(u_{3}
+3u_{2})-6u_{3}\frac{\partial u_{2}}{\partial x}-3u_{1}\frac{\partial u_{1}}{\partial x}
-3u_{2}\frac{\partial u_{2}}{\partial x}-2u_{3}^{3}+3u_{3}(u_{1}^{2}+u_{2}^{2}) \biggr) .
\end{aligned}
\end{equation}

These equations acquire  Hamiltonian form:
\begin{equation}
\label{Hamilton}
\frac{\partial u_{i}}{\partial t}=\frac{\partial}{\partial x} \left( \frac{\delta H}{\delta u_{i}} \right) = \frac{\partial v_{0;i}}{ \partial x }
\end{equation}
where the Hamiltonian is:
\begin{equation}\label{eq:H22}
\begin{aligned}
H=& \int_{-\infty}^{\infty} dx \biggl( -\frac{1}{2}\sum_{i=1}^{3}u_{i}^{4}+\frac{3}{2} \underset{i<j}{\sum_{i=1}^{3}
\sum_{j=1}^{3}}u_{i}^{2}u_{j}^{2}+\frac{5}{2} \left( \frac{\partial u_{1}}{\partial x} \right)^{2}
-2 \left( \frac{\partial u_{2}}{\partial x} \right)^{2} \\
& -\frac{1}{2} \left( \frac{\partial u_{3}}{\partial x} \right)^{2}+
+\frac{\partial u_{2}}{\partial x} \left( \frac{9}{2}u_{1}^{2}-3u_{3}^{2} \right) +\frac{3}{2}\frac{\partial u_{3}}{\partial x}(u_{1}^{2}
+u_{2}^{2})-3 \left( \frac{\partial u_{2}}{\partial x} \right) \left( \frac{\partial u_{3}}{\partial x} \right) \biggr) .
\end{aligned}
\end{equation}
Using the second type of grading is in fact equivalent to the first one. One can check that the two types of gradings are related by a similarity transformations of the form:
\begin{equation}\label{eq:simil}
w_0^{-1}\tilde{Q}w_0 = Q, \qquad  w_0^{-1}\tilde{J}w_0 = J. \qquad
w_0 = \frac{1}{\sqrt{5}}\left(\begin{array}{cccccc}
 \frac{1}{2} & \frac{1}{2} & \frac{1}{2} & i\frac{\sqrt{5}}{2} &\frac{1}{2} & \frac{1}{2} \\
\omega_2 ^{-4}  & \omega_2 ^{-2}  & 1 & 0 & \omega_2 ^{4} & \omega_2 ^{2}  \\
 \omega_2 ^{2}  & \omega_2 ^{-4}  & 1 & 0 & \omega_2 ^{-2} & \omega_2 ^{4} \\
\omega_2 ^{-2}  & \omega_2 ^{4}  & 1 & 0 & \omega_2 ^{2} & \omega_2 ^{-4} \\
 \omega_2 ^{4}  & \omega_2 ^{2}  & 1 & 0 & \omega_2 ^{-4} & \omega_2 ^{-2} \\
1  & 1 & 1 & -i\sqrt{5} & 1 & 1  \end{array}\right).
\end{equation}
Effectively we find that $u_j$ and $q_s$ are related linearly as follows:
\begin{equation}\label{eq:u-q123}\begin{aligned}
u_1 =&  c^- q_1 + c^+ q_3, &\quad u_2 =&  \frac{1}{\sqrt{5}} q_2 , &\quad u_3 =&  -c^- q_1 + c^+ q_3, \\ c^+ =&  \frac{\sqrt{10 + 2 \sqrt{5}} }{10}, & \quad c^- =&  \frac{\sqrt{10 - 2 \sqrt{5}} }{10} .
\end{aligned}\end{equation}

\subsection{ The recursion relations and the recursion operators $\Lambda_k$}

Our aim here is to describe the  hierarchies of equations in terms of the recursion operators $\Lambda_k$. The idea is to look upon the compatibility conditions  as recurrent relations which will be solved using the recursion operators, see \cite{GSIB, GMSV, GMSV4, GMSV3, GMSV5}. The initial condition is provided by:
\begin{equation}\label{eq:V2Ic}\begin{aligned}
 V_2 = \ad_J^{-1} [K,Q]
\end{aligned}\end{equation}
Note that the operator $\ad_J$ acting on any element $X\in \mathfrak{g}$ by $\ad_J X = [J,X]$ has a kernel; therefore it could be inverted only if $X$ belongs to its image. Therefore in solving the recurrent relations we will need to split each of the $V_s$ into, roughly speaking, `diagonal` and `off-diagonal` parts:
\begin{equation}\label{eq:Vsfd}\begin{aligned}
 V_s = V_s^{\rm f} + V_s^{\rm d},
\end{aligned}\end{equation}
where $V_s^{\rm f} \in \im \ad_J $ and $V_s^{\rm d} $ is such that $\ad_J V_s^{\rm f} =0 $.

 \begin{equation}\label{eq:Vnm1}\begin{aligned}
  V_{n-1}(x,t) \equiv V_{n-1}^{\rm f}(x,t) = \sum_{p=1}^{r} \frac{\alpha_p(K) }{\alpha_p(J)} q_p(x,t) \mathcal{E}_p^{(n_1-1)}.
 \end{aligned}\end{equation}

Now let us assume that $s_1$ is an exponent and split the third equation in (\ref{eq:LM0}) into diagonal and off-diagonal parts.
Evaluating the Killing form of this equation with $\mathcal{H}_1^{h-s_1}$ we obtain:
\begin{equation}\label{eq:ws}\begin{aligned}
 w_s(x,t) = \frac{i}{c_{s_1}} \partial_x^{-1} \left\langle [Q,V_s^{\rm f}] ,\mathcal{H}_1^{h-s_1} \right\rangle + \const, \qquad
 c_{s_1} = \left\langle \mathcal{H}_1^{s_1}, \mathcal{H}_1^{h-s_1}\right\rangle .
\end{aligned}\end{equation}
In what follows for simplicity we will set all these integration constants to 0. A diligent reader can easily work out the more
general cases when some of these constants do not vanish.
The off-diagonal part of the third equation in (\ref{eq:LM0}) gives:
\begin{equation}\label{eq:Vsm1f}\begin{aligned}
i \partial_x V_s^{\rm f} + [Q, V_s^{\rm f}]^{\rm f} + [Q, w_s \mathcal{H}_1^{s_1}] = [J, V_{s-1}],
\end{aligned}\end{equation}
i.e.
\begin{equation}\label{eq:Vsm1f2}\begin{aligned}
 V_{s-1}^{\rm f} = \ad_J^{-1} \left( i \partial_x V_s^{\rm f} + [Q, V_s^{\rm f}]^{\rm f} + [Q, w_s \mathcal{H}_1^{s_1}]\right) = \Lambda_{s_1} V_s^{\rm f}.
\end{aligned}\end{equation}
Thus we obtained the integro-differential operator $\Lambda_{s_1} $ which acts on any $Z \equiv Z^{\rm f} \in \mathfrak{g}^{(s_1)}$ by:
\begin{equation}\label{eq:Lams1}\begin{aligned}
 \Lambda_{s_1} Z = \ad_J^{-1} \left( i \partial_x Z + [Q, Z]^{\rm f} + \frac{i}{c_{s_1}}[Q,  \mathcal{H}_1^{s_1}]
 \partial_x^{-1} \left\langle [Q,Z] ,\mathcal{H}_1^{h-s_1} \right\rangle  \right).
\end{aligned}\end{equation}
If $s_1$ is not an exponent we have only to work out the off-diagonal part of the third equation in (\ref{eq:LM0})
with the result:
\begin{equation}\label{eq:Vsm1f3}\begin{aligned}
 V_{s-1}^{\rm f} =&  \ad_J^{-1} \left( i \partial_x V_s^{\rm f} + [Q, V_s^{\rm f}]^{\rm f} \right) = \Lambda_{0} V_s^{\rm f}, \\
 \Lambda_0 Z =&  \ad_J^{-1} \left( i \partial_x Z + [Q, Z]^{\rm f} \right).
\end{aligned}\end{equation}
Now $\Lambda_0$ is a differential operator.

Now we can treat the hierarchies related to $A_5^{(1)}$. Since the Coxeter number is 6 and the exponents are  $1, 2, 3, 4, 5$ the results are:
\begin{equation}\label{eq:A51Mkdv}\begin{aligned}
n=&  6n_0 +1 &\quad \partial_t Q =&  \partial_x \left(\bm{\Lambda}^{n_0} Q(x,t)\right) ,&\quad  f(\lambda) =&  \lambda^{N_1} \mathcal{H}_1^{(1)}, \\
n=&  6n_0 +a &\quad \partial_t Q =&  \partial_x \left( \bm{\Lambda}^{n_0} \Lambda_{a-1} \dots  \Lambda_0 \ad_J^{-1} [\mathcal{H}_1^{a}, Q(x,t)] \right),
&\quad f(\lambda) =&  \lambda^{N_a}\mathcal{H}_1^{(a)},
\end{aligned}\end{equation}
where $N_a=6n_0+a$, $a=1, 2,\dots, 5$ and $\bm{\Lambda}=\Lambda_1 \Lambda_2 \Lambda_3\Lambda_4 \Lambda_5\Lambda_0$.

Similarly we can treat the hierarchies related to $A_5^{(2)}$. Here the Coxeter number is 10 and the exponents are $1, 3, 5, 7, 9$. The results are:
\begin{equation}\label{eq:A52Mkdv}\begin{aligned}
n=&  10n_0 +1 &\quad \partial_t Q =&  \partial_x \left(\bm{\Lambda}^{n_0} Q(x,t)\right) ,\\
n=&  10n_0 +3 &\quad \partial_t Q =&  \partial_x \left( \bm{\Lambda}^{n_0} \Lambda_1 \Lambda_0 \ad_J^{-1} [\mathcal{H}_1^{(3)}, Q(x,t)] \right), \\
n=&  10n_0 +5 &\quad \partial_t Q =&  \partial_x \left(\bm{\Lambda}^{n_0} \Lambda_1 \Lambda_0 \Lambda_3 \Lambda_0  \ad_J^{-1}[\mathcal{H}_1^{(5)}, Q(x,t)]\right) ,  \\
n=&  10n_0 +7 &\quad \partial_t Q =&  \partial_x \left(\bm{\Lambda}^{n_0} \Lambda_1 \Lambda_0 \Lambda_3 \Lambda_0 \Lambda_5 \Lambda_0  \ad_J^{-1}[\mathcal{H}_1^{(7)}, Q(x,t)]\right) , \\
n=&  10n_0 +9 &\quad \partial_t Q =&  \partial_x \left(\bm{\Lambda}^{n_0} \Lambda_1 \Lambda_0 \Lambda_3 \Lambda_0 \Lambda_5 \Lambda_0 \Lambda_7 \Lambda_0 \ad_J^{-1}[\mathcal{H}_1^{(9)}, Q(x,t)]\right) ,
\end{aligned}\end{equation}
where $\bm{\Lambda}=\Lambda_1 \Lambda_0 \Lambda_3\Lambda_0 \Lambda_5\Lambda_0 \Lambda_7\Lambda_0 \Lambda_9\Lambda_0$ and the dispersion laws are given by $f_j(\lambda) = \lambda^{10n_0+n_j} \mathcal{H}_{n_j}^{(1)}$, $n_j = 2j-1$ being the exponents of $A_5^{(2)}$.

\section{The Riemann-Hilbert problem}

\subsection{General aspects}

The general methods for constructing the FAS of the Lax operators started with the pioneer papers by A. B. Shabat \cite{Sha*75, Sha} in which he constructed the FAS of a class of $n\times n$ Lax operators of the type (\ref{eq:A51}) in which $J = \diag (a_1, \dots, a_n) $ assuming that the eigenvalues of $J$ are real and are ordered $a_k > a_j$ if $k<j$. The continuous spectrum of such $L$ operator with rapidly vanishing potential $Q$ fills up the real axis in the complex $\lambda$-plane. One of the corresponding FAS $\chi^+(x,\lambda)$ allowed analytic extension into the upper half plane $\mathbb{C}_+$; the other one $\chi^-(x,\lambda)$ is analytic in the lower half plane $\mathbb{C}_-$ and on the real axis they are related linearly:
\begin{equation}\label{eq:rhp0}\begin{aligned}
\chi^+(x,t,\lambda) = \chi^-(x,t,\lambda) G_0(t,\lambda),
\end{aligned}\end{equation}
where the sewing function $G(t,\lambda)$ is expressed by the Gauss factors of the corresponding scattering matrix. A simple transformation from $\chi^\pm(x,\lambda)$ to $\xi^\pm(x,\lambda) = \chi^\pm(x,\lambda)e^{i\lambda Jx}$ allows one to reformulate the RHP (\ref{eq:rhp0}) as:
\begin{equation}\label{eq:rhp1}\begin{aligned}
\xi^+(x,t,\lambda) = \xi^-(x,t,\lambda) G(x,t,\lambda), \qquad G(x,t,\lambda) = e^{-i\lambda Jx}  G_0(t,\lambda)  e^{i\lambda Jx} .
\end{aligned}\end{equation}
The RHP (\ref{eq:rhp1}) has the advantage that it allows canonical normalization in the form $\lim_{\lambda \to \infty}\xi^\pm(x,t,\lambda) =\openone $.

Shabat and Zakharov developed further these ideas by discovering the deep relation between the RHP (\ref{eq:rhp1}) and the corresponding pair of Lax operators. They proved a theorem \cite{ZaSha1, ZaSha2} stating that if $\xi^\pm(x,t,\lambda) $ satisfy the RHP (\ref{eq:rhp1}) and the sewing function $G(x,t,\lambda) $ has proper $x$-dependence, then the corresponding $\chi^\pm (x,t,\lambda)$  will be FAS of the relevant Lax pair.

The next important step that they proposed was to devise a method of deriving a special class of singular solutions to the RHP. Today it is known as the Zakharov-Shabat dressing method \cite{ZaSha1, ZaSha2, ZMNP}. It has several formulations and is one of the best known methods for constructing the multi-soliton solutions of the integrable NLEE.
Later Shabat's results were generalized to the class of Lax operators whose potentials $Q$ and $J$ take values in simple Lie algebras $\mathfrak{g}$ \cite{Ge}.

The next important step in this direction was taken by Beals and Coifman \cite{BC} who treated the general case of $n\times n$ Lax operators with complex-valued $J$. The substantial difference with the Shabat's case was that the continuous spectrum of $L$ fills up a set of rays $l_p$, which split the complex $\lambda$-plane $\mathbb{C}$ into several sectors $\Omega_p$.
In each of these sectors Beals and Coifman succeeded to construct FAS $\xi_p(x,\lambda)$.
Let us assume that the sectors $\Omega_p$ and $\Omega_s$ share the ray $l_p$, then we can have a set of relations like
\begin{equation}\label{eq:rhp2}\begin{aligned}
\xi_p(x,t,\lambda) = \xi_{s}(x,t,\lambda) G_p(x,t,\lambda), \quad G_p(x,t,\lambda) = e^{-i\lambda Jx}  G_{p0}(t,\lambda)  e^{i\lambda Jx} ,
\end{aligned}\end{equation}
where $l_p = \Omega_p\cap \Omega_s$, $p=1,2, \dots$, which will be a generalized RHP. Zakharov-Shabat theorem mentioned above and the dressing method can easily be extended to such generalized RHP. And of course the results of Beals and Coifman were generalized also to the case when $Q(x,t)$ and $J$ take values in any simple Lie  algebra $\mathfrak{g}$ \cite{VG-Ya-13, GYa, SIAM*14}.

Let us also mention briefly how the analyticity properties of $\xi_\nu (x,t, \lambda)$ are proved. Since $\chi_\nu (x,t, \lambda)$ are fundamental solutions of the operators $L$ and $M$ above then $\xi_\nu (x,t, \lambda)$ will be fundamental solutions of the related operators:
\begin{equation}\label{eq:eqXi}\begin{aligned}
 \tilde{L}\chi_\nu \equiv i \frac{\partial \xi_\nu}{ \partial x } &+ Q(x,t)\xi_\nu (x,t, \lambda) - \lambda [J, \xi_\nu] =0, \\
 \tilde{M}\chi_\nu \equiv  i \frac{\partial \xi_\nu}{ \partial t } &+ V(x,t,\lambda)\xi_\nu (x,t, \lambda) - \lambda^3 [K, \xi_\nu ] =0, \quad  V(x,t,\lambda) = \sum_{p=0}^{2} V_p(x,t) \lambda^p.
\end{aligned}\end{equation}
We already made special choices for both $Q(x,t)$ and $J$ using two different specific gradings of $A_5 \simeq sl(6)$. Each of these choices can be viewed as a realization of Mikhailov reduction group $\mathbb{Z}_h$ \cite{Mikh}:
\begin{equation}\label{eq:Zh0}\begin{aligned}
C(Q(x,t) - \lambda J) =&  Q(x,t) - \lambda \omega J, \\ C(V(x,t,\lambda) - \lambda^3 K) =&  V(x,t, \lambda \omega) - \lambda^3 \omega^3 K,
\end{aligned}\end{equation}
with properly chosen Coxeter automorphism $C$ such that $C^h =\openone$, and $h$ is the Coxeter number. In other words  the Lax pairs with $\mathbb{Z}_h$  reductions of Mikhailov type \cite{Mikh} provide an important class of Lax operators with complex-valued $J$. It is also natural to remember that in fact the potentials $Q(x,t) - \lambda J$ and $V(x,t,\lambda) - \lambda^3 K$ of these Lax pairs take values in a Kac-Moody algebras, which are based on the simple Lie algebras graded by Coxeter automorphisms \cite{Bourb, DS0, DS, Kac, Cart}.

The derivation of the FAS of eq. (\ref{eq:eqXi}) is based on the set of integral equations which incorporate also the asymptotic behavior of $\xi_\nu (x,t, \lambda)$ for $x\to \pm \infty$. These equations have the form (see \cite{BC, VG-Ya-13, GYa, SIAM*14}):
\begin{equation}\label{eq:Inteq}\begin{aligned}
( \xi_\nu (x,t, \lambda) )_{kj} =&  \delta_{kj} + i \int_{-\infty}^{x} dy\; \left( Q(y,t) \xi_\nu (y,t, \lambda)\right)_{kj} e^{-i\lambda (J_k - J_j)(x-y)} , \\
  & \qquad \qquad \mbox{for} \quad \lambda \in \Omega_\nu \quad \mbox{and}\quad \im \lambda (J_k - J_j) \mathop{\leq}\limits_{\nu} 0,\\
( \xi_\nu (x,t, \lambda) )_{kj} =&   i \int_{\infty}^{x} dy\; \left( Q(y,t) \xi_\nu (y,t, \lambda)\right)_{kj} e^{-i\lambda (J_k - J_j)(x-y)} , \\
  &\qquad \qquad  \mbox{for} \quad \lambda \in \Omega_\nu \quad \mbox{and}\quad \im \lambda (J_k - J_j) \mathop{>}\limits_{\nu} 0,
\end{aligned}\end{equation}
where the index $\nu$ in the inequalities in eq. (\ref{eq:Inteq}) means that we must restrict $\lambda \in \Omega_\nu$.

Roughly speaking our first task in analyzing the integral equations (\ref{eq:Inteq}) must be to determine those lines in the complex $\lambda$-plane on which the exponential factors in the integrands oscillate. Normally these lines will constitute the continuous spectrum of $\tilde{L}$. They would be determined by $\im \lambda (J_k - J_j) =0$, which can be written in the form:
\begin{equation}\label{eq:alJ}\begin{aligned}
 \im \lambda \alpha(J) =0,
\end{aligned}\end{equation}
where $\alpha = e_k - e_j$ is a root of $A_5$.  The set of equations (\ref{eq:alJ}) where $\alpha$ runs over the root system $\Delta$ of $A_5$ are simple algebraic equations. Their solutions are collected in Table \ref{tab:2} for $A_5^{(1)}$ and in Table \ref{tab:4} for $A_5^{(2)}$. Thus we establish that the continuous spectrum of $\tilde{L}$ fills up all rays $l_\nu \equiv \arg \lambda = \nu \pi/h$, $\nu= 0,1,\dots, 2h-1$.

\begin{lemma}\label{lem:1}
To each pair of rays $l_\nu \cup l_{2h-\nu}$ there corresponds a subalgebra $\mathfrak{g}_\nu \subset sl(6)$ which in the case of $A_5^{(1)}$ is isomorphic either to $sl(2)\oplus sl(2)$, or to $sl(2)\oplus sl(2) \oplus sl(2)$. In the case of $A_5^{(2)}$ it is isomorphic either to $sl(2)\oplus sl(2)$, or to $sl(2) $.
\end{lemma}
\begin{proof}
 It is obvious that if $\alpha$ is a solution to eq. (\ref{eq:alJ}), then $-\alpha$ will also be a solution. It remains to check that any two non-proportional roots  related to each pair of rays  $l_\nu \cup l_{2h-\nu}$ are mutually orthogonal. Thus inspecting Table \ref{tab:2} we prove the lemma for $A_5^{(1)}$. Similarly, inspecting Table \ref{tab:4} we prove the lemma for $A_5^{(2)}$.
\end{proof}

\begin{theorem}\label{thm:1}
 The solution $\xi_\nu (x,t,\lambda)$ of eq. (\ref{eq:Inteq}) is an analytic function of $\lambda$ for $\lambda \in \Omega_\nu$. In addition:
 \begin{equation}\label{eq:RC}\begin{aligned}
 C(\xi_\nu (x,t,\lambda)) = \xi_{\nu+2} (x,t,\lambda \omega).
 \end{aligned}\end{equation}
\end{theorem}

\begin{proof} [Idea of the proof]
The solutions of the conditions $  \im \lambda (J_k - J_j) \mathop{\leq}\limits_{\nu} 0$ for $\lambda \in \Omega_\nu $ in the case of $A_5^{(1)}$ are listed in Table \ref{tab:3} as the subsets  $\delta_\nu^+$. All other roots of $A_5$ for $\lambda \in \Omega_\nu$ will satisfy the condition
$  \im \lambda (J_k - J_j) \mathop{>}\limits_{\nu} 0$. As a result it is easy to see that the exponential factors in eq. (\ref{eq:Inteq}) will decrease exponentially for all $x$ and $\lambda \in \Omega_p$. In particular this means that the integrals will be convergent for any $\lambda \in \Omega_\nu$, which guarantees the existence of $\xi_\nu (x,t,\lambda)$. Let us now consider the integral equations for the derivatives $\frac{\partial^s }{ \partial \lambda^s} \xi_\nu (x,t,\lambda)$. The integrands of these equations will contain, besides the exponential factors, also polynomial factors in $x$ and $y$ of order $s$. Again the decreasing exponential factors ensure the convergence of the integrals in the right hand side, which means  that $\xi_\nu (x,t,\lambda)$ will allow derivatives of all orders with respect to $\lambda$ in the sector $\Omega_\nu$. This is one of the basic properties of the analytic functions.

Finally, equation (\ref{eq:RC}) follows directly from Mikhailov reduction condition (\ref{eq:Zh0}).
\end{proof}

The corresponding generalized RHP can be written down as follows:
\begin{equation}\label{eq:rhp2a}\begin{aligned}
\xi_{\nu}(x,t,\lambda) = \xi_{\nu -1}(x,t,\lambda) G_\nu(x,t,\lambda), \qquad G_\nu(x,t,\lambda) = e^{-i\lambda Jx}  G_{\nu 0}(t,\lambda)  e^{i\lambda Jx} ,
\end{aligned}\end{equation}
where $ \lambda \in l_\nu $ and the rays $l_\nu$ are determined by $\arg \lambda = \nu \pi /h$, $\nu= 0, \dots 2h-1$ where $h$ is the Coxeter number. The sector $\Omega_\nu $ is determined by the rays $l_\nu$ and $l_{\nu+1}$, see Figure \ref{fig:2}. In fact A. V. Mikhailov, developing his ideas on the reduction groups in \cite{Mikh},  came very close to such formulation of the RHP.

\begin{figure}
  \includegraphics[width=0.45\textwidth]{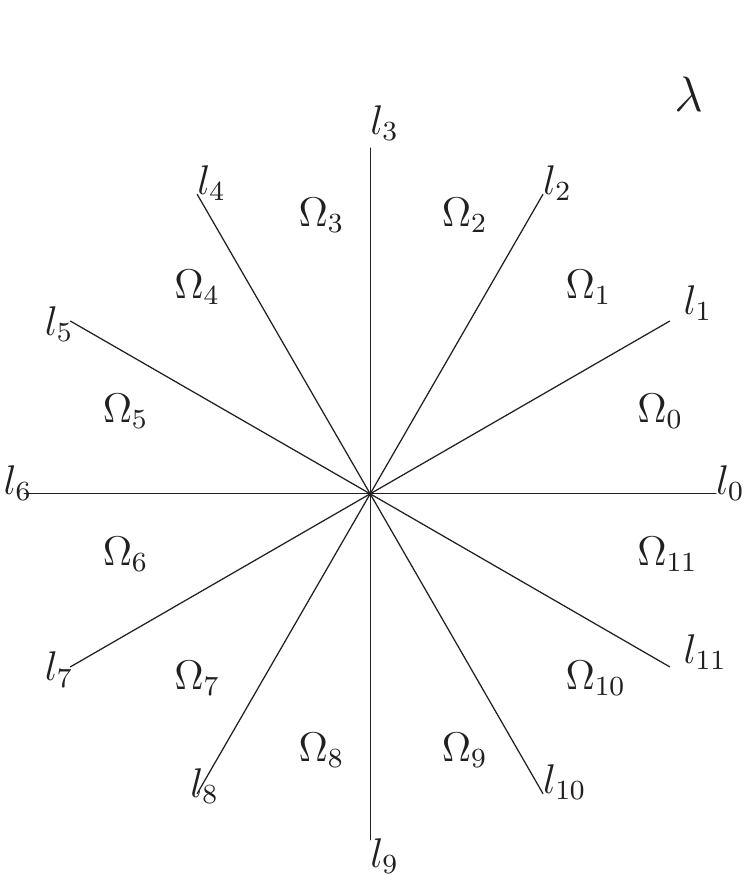}\quad
    \includegraphics[width=0.45\textwidth]{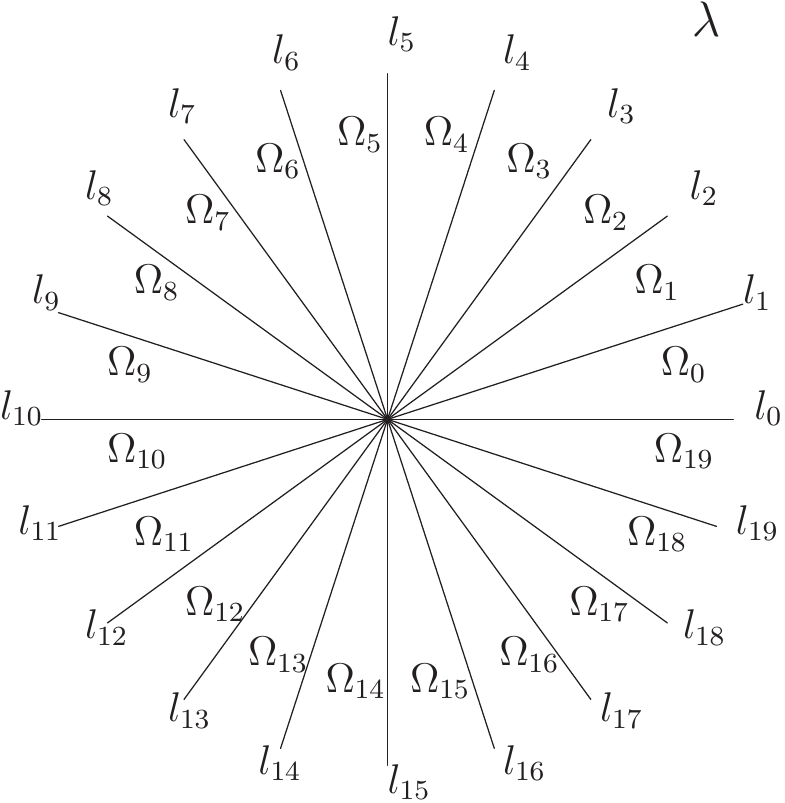}\\
  \caption{Continuous spectrum of the Lax operators and contours of the RHP for  $A_5^{(1)}$ (left panel) and  $A_5^{(2)}$ (right panel).}\label{fig:2}
\end{figure}

\begin{table}
  \centering
  \begin{tabular}{|c|c|c|}
    \hline
    $l_\nu$ & $l_0 \cup l_6$ & $l_1 \cup l_7$ \\
    $\alpha$ & $\pm (e_1-e_4), \pm (e_2-e_3), \pm (e_5-e_6)$ & $\pm (e_1-e_3), \pm (e_4-e_6)$ \\
    \hline
    $l_\nu$ & $l_2 \cup l_8$ &  $l_3 \cup l_9$ \\
    $\alpha$ &  $\pm (e_1-e_2), \pm (e_3-e_6), \pm (e_4-e_5)$ & $\pm (e_2-e_6), \pm (e_3-e_5)$ \\
    \hline
    $l_\nu$ & $l_4 \cup l_{10}$ &  $l_5 \cup l_{11}$ \\
    $\alpha$ &  $\pm (e_1-e_6), \pm (e_2-e_5), \pm (e_3-e_4)$ & $\pm (e_1-e_5), \pm (e_2-e_4)$ \\
      \hline
  \end{tabular}
  \medskip
  \caption{The roots of $A_5^{(1)}$ related to the rays $l_\nu$, $\nu=0, \dots, 11$, see the left panel of Figure \ref{fig:2} }\label{tab:2}
\end{table}

\begin{table}
  \centering
  \begin{tabular}{|c|c|c|}
    \hline
    $\Omega_\nu$ & $\delta_\nu^+ $ & $\delta_\nu^-$ \\ \hline
    $\Omega_0$ & $ (e_1-e_4), (e_2-e_3), - (e_5-e_6)$ & $(e_1-e_5),  (e_2-e_4)$ \\
    \hline
    $\Omega_1$ & $ (e_1-e_5),  (e_2-e_4) $ & $ (e_1-e_6), (e_2-e_5), (e_3-e_4)$ \\
    \hline
    $\Omega_2$ & $ (e_1-e_6), (e_2-e_5),  (e_3-e_4)$ & $(e_2-e_6),  (e_3-e_5)$ \\
    \hline
    $\Omega_3$ & $ (e_2-e_6),  (e_3-e_5) $ & $ -(e_1-e_2), (e_3-e_6), (e_4-e_5)$ \\
    \hline
    $\Omega_4$ & $ -(e_1-e_2), (e_3-e_6),  (e_4-e_5)$ & $-(e_1-e_3),  (e_4-e_6)$ \\
    \hline
    $\Omega_5$ & $ -(e_2-e_3),  (e_4-e_6) $ & $ -(e_2-e_3), (e_5-e_6), -(e_1-e_4)$ \\
      \hline
  \end{tabular}
    \medskip
  \caption{The root subsystems $\delta_\nu^\pm $ of $A_5^{(1)}$ related to the sectors $\Omega_\nu$, $\nu=0, \dots, 11$, see the left panel of Figure \ref{fig:2} }\label{tab:3}
\end{table}

\begin{table}
  \centering
  \begin{tabular}{|c|c|c|c|c|}
    \hline
    $l_\nu$ & $l_0 \cup l_{10}$ & $l_1 \cup l_{11}$ & $l_2 \cup l_{12}$ & $l_3 \cup l_{13}$ \\
    $\alpha$ & $\pm (e_3-e_4) $ & $\pm (e_1-e_2), \pm (e_3-e_5)$ &
    $\pm (e_4-e_5) $ & $\pm (e_2-e_5), \pm (e_3-e_6)$ \\ \hline
    $l_\nu$ & $l_4 \cup l_{14}$ & $l_5 \cup l_{15}$ & $l_6 \cup l_{16}$ & $l_7 \cup l_{17}$ \\
    $\alpha$ & $\pm (e_2-e_4) $ & $\pm (e_1-e_5), \pm (e_2-e_6)$ &
    $\pm (e_4-e_6) $ & $\pm (e_1-e_6), \pm (e_2-e_3)$ \\ \hline
    $l_\nu$ & $l_8 \cup l_{18}$ & $l_9 \cup l_{19}$ &  &  \\
    $\alpha$ & $\pm (e_1-e_4) $ & $\pm (e_1-e_3), \pm (e_5-e_6)$ &
     &  \\
      \hline
  \end{tabular}
    \medskip
  \caption{The roots of $A_5$ related to the rays $l_\nu$, $\nu=0, \dots, 19$ with $J= \diag(\omega_2, \omega_2^3, -1, 0, \omega_2^9, \omega_2^7) $, see the right panel of Figure \ref{fig:2} and Remark \ref{rem:2}. }\label{tab:4}
\end{table}

\begin{remark}\label{rem:2}
 For technical reasons in Tables \ref{tab:4} and \ref{tab:5} we list the roots of $A_5$. Their root vectors $E_{ij}$  can easily be expressed in terms of the root vectors of $A_5^{(2)}$ taking into account the relations (\ref{eq:V2b}) from Appendix B. Indeed,
 \begin{equation}\label{eq:Eij}\begin{aligned}
  E_{ij} =&  \frac{1}{2} (\tilde{\mathcal{E}}^+_{ij}+ \tilde{\mathcal{E}}^-_{ij}), \qquad
 E_{i\bar{j}} =&  \frac{1}{2} (\tilde{\mathcal{E}}^+_{i\bar{j}}- \tilde{\mathcal{E}}^-_{i\bar{j}}), \qquad E_{j\bar{j}} =&  \tilde{\mathcal{E}}^+_{j\bar{j}} ,
 \end{aligned}\end{equation}
where $1 \leq i <j \leq 3$ and $\bar{k} = 7-k$.
\end{remark}

\begin{table}
  \centering
  \begin{tabular}{|c|c|c||c|c|c|}
    \hline
    $\Omega_\nu$ & $\delta_\nu^+ $ & $\delta_\nu^-$ & $\Omega_\nu$ & $\delta_\nu^+ $ & $\delta_\nu^-$ \\ \hline
    $\Omega_0$ & $ (e_1-e_4)$ & $ -(e_1-e_3), - (e_5-e_6)$ & $\Omega_1$ & $(e_1-e_5),  (e_2-e_4)$ & $- (e_1 -e_4)$ \\    \hline
    $\Omega_2$ & $ -(e_1-e_4)$ & $ -(e_1-e_6), - (e_2-e_3)$ & $\Omega_3$ & $-(e_1-e_6), -(e_2-e_3)$ & $- (e_4 -e_6)$ \\     \hline
    $\Omega_4$ & $ -(e_4-e_6)$ & $ -(e_1-e_5), - (e_2-e_6)$ & $\Omega_5$ & $-(e_1-e_5), -(e_2-e_6)$ & $- (e_2 -e_4)$ \\
    $\Omega_6$ & $ -(e_2-e_4)$ & $ -(e_2-e_5), - (e_3-e_6)$ & $\Omega_7$ & $-(e_2-e_5), -(e_3-e_6)$ & $- (e_4 -e_5)$ \\ \hline
    $\Omega_8$ & $ -(e_4-e_5)$ & $ (e_1-e_2), - (e_3-e_5)$ & $\Omega_9$ & $(e_1-e_2), -(e_3-e_5)$ & $- (e_3 -e_4)$ \\ \hline
  \end{tabular}
    \medskip
  \caption{The root subsystems $\delta_\nu^\pm $ of $A_5$ related to the sectors $\Omega_\nu$, $\nu=0, \dots, 9$, see the left panel of Figure \ref{fig:2} and Remark \ref{rem:2}. }\label{tab:5}
\end{table}

It is obvious that all the information about the scattering data of $L$ (or $\tilde{L}$) must be hidden in the sewing functions $G_\nu (x,t,\lambda)$. For the Lax operators we are considering it is not possible to introduce Jost solutions without imposing additional severe restrictions on $Q(x,t)$, such as tending to 0 for $x \to \pm \infty$ faster tan any exponential $e^{-c|x|}$ for any positive $c$, or even assuming that $Q(x,t)$ has finite support. However, we can use the limits of $\xi_\nu (x,t,\lambda)$ for $x \to \pm \infty$ and $\lambda \in l_\nu$.
They are given by \cite{GYa, SIAM*14}:
\begin{equation}\label{eq:limchi}\begin{aligned}
\lim_{x\to -\infty} e^{i\lambda J x} \chi_\nu (x,t,\lambda) =&  S_\nu^+(t,\lambda), &\quad \lambda &\in l_\nu e^{ i0} , \\
\lim_{x\to \infty} e^{i\lambda J x} \chi_\nu (x,t,\lambda) =&  T_\nu^-(t,\lambda) D_\nu^+(\lambda), &\quad \lambda &\in l_\nu e^{ i0} , \\
\lim_{x\to -\infty} e^{i\lambda J x} \chi_{\nu-1} (x,t,\lambda) =&  S_\nu^-(t,\lambda), &\quad \lambda &\in l_\nu e^{-i0} , \\
\lim_{x\to \infty} e^{i\lambda J x} \chi_{\nu -1} (x,t,\lambda) =&  T_\nu^+(t, \lambda) D_\nu^-(\lambda), &\quad \lambda &\in l_\nu e^{-i0} ,
\end{aligned}\end{equation}
where  $\nu =0,1,\dots, 2h-1$ and $S_\nu^\pm $, $T_\nu^\pm $ and $D_\nu^\pm $ have the form:
\begin{equation}\label{eq:STDpm}\begin{aligned}
S_\nu^\pm (\lambda) =&  \exp \left( \sum_{\alpha \in \delta_\nu^+}^{} s_\alpha^\pm (\lambda) E_{\pm \alpha} \right), & \quad T_\nu^\pm (\lambda) =&  \exp \left( \sum_{\alpha \in \delta_\nu^+}^{} \tau_\alpha^\pm (\lambda) E_{\pm \alpha} \right), \\
D_\nu^\pm (\lambda) =&  \exp \left( \sum_{\alpha \in \delta_\nu^+}^{} d_{\nu, \alpha}^\pm (\lambda) H_{ \alpha} \right),
\end{aligned}\end{equation}

\begin{remark}\label{rem:1}
Formally one can introduce an analogue of the scattering matrix for each pair of rays $l_\nu \cup l_{h+\nu}$ as follows:
\begin{equation}\label{eq:Tnu}\begin{aligned}
 T_\nu(t,\lambda) = T_\nu^-(t,\lambda) D_\nu^+(\lambda) \hat{S}^+_\nu (t,\lambda) = T_\nu^+(t,\lambda) D_\nu^-(\lambda) \hat{S}^-_\nu (t,\lambda), \qquad \lambda \in l_\nu.
\end{aligned}\end{equation}
Note that $T_\nu (t,\lambda)$ belongs to the subgroup $\mathcal{G}_\nu \subset SL(6)$ whose root system is $\delta_\nu^+ \cup \delta_\nu^-$. Then $T_\nu^\pm (t,\lambda)$, $S_\nu^\pm (t,\lambda)$ and $D_\nu^\pm (\lambda)$ can be viewed as the Gauss factors of $T_\nu (t,\lambda)$. Another peculiar fact here is, that to each sector $\Omega_\nu$ we relate specific ordering of the root systems, i.e. specific choice of the positive and negative roots, see \cite{GYa, SIAM*14}.

\end{remark}

\begin{lemma}\label{lem:2}
i) The $t$-dependence of the scattering data for the mKdV equations is given by:
\begin{equation}\label{eq:dSpmt}\begin{aligned}
i \frac{\partial T_\nu^\pm }{ \partial t} - \lambda^3 [K, T_\nu^\pm (t,\lambda)] =& 0, & \qquad
i \frac{\partial S_\nu^\pm }{ \partial t} - \lambda^3 [K, S_\nu^\pm (t,\lambda)] =& 0, \\
i \frac{\partial D_\nu^\pm }{ \partial t} =& 0, &\qquad
i \frac{\partial T_\nu }{ \partial t} - \lambda^3 [K, T_\nu (t,\lambda)] =& 0.
\end{aligned}\end{equation}
ii) The function $ D_\nu^+(\lambda)$ (resp. $ D_\nu^-(\lambda)$) is an analytic function for
$\lambda \in \Omega_\nu$ (resp. for $\lambda \in \Omega_{\nu-1}$). They are generating functionals of the integrals of motion for the mKdV hierarchy.
\end{lemma}

\begin{proof}
i) Let us multiply the second equation in (\ref{eq:eqXi}) by $e^{i \lambda J x}$ and take the limits for $x \to \infty$ and $x \to -\infty$. Takin into account eq. (\ref{eq:limchi}) and the fact, that $Q(x,t)$ and $V(x,t,\lambda)$ vanish fast enough for $x\to \pm \infty$ we easily obtain the equations (\ref{eq:dSpmt}).

\medskip
ii) The analyticity properties of $D_\nu^\pm (\lambda)$ were proven in \cite{SIAM*14} for generic Kac-Moody algebras. As generating functionals of the integrals of motion it is more convenient to consider $d_{\nu,\alpha}^\pm (\lambda)$. Their asymptotic expansions:
\begin{equation}\label{eq:dnu}\begin{aligned}
 d_{\nu,\alpha}^\pm (\lambda) = \sum_{p=1}^{\infty} \lambda^{-p} I_{\nu, \alpha}^{(p)}
\end{aligned}\end{equation}
provide integrals of motion $I_{\nu, \alpha}^{(p)}$ whose densities are local in $Q(x,t)$, i.e. depend only on $Q(x,t)$ and its $x$-derivatives.

\end{proof}

\section{The minimal set of scattering data}

Here we reformulate the basic results of \cite{GYa, SIAM*14} for the specific Kac-Moody algebras used above. It is natural to expect that these sets will be expressed in terms of the sewing functions of the RHP. Our considerations will be relevant only for the cases when the solution of the RHP is regular. This means that the spectra of the corresponding Lax operators do not contain discrete eigenvalues.

\subsection{The $A_5^{(1)}$ case}

We introduce two minimal sets of scattering data for the $A_5^{(1)}$ Kac-Moody algebra as follows, see Table \ref{tab:3}:
\begin{equation}\label{eq:T12}\begin{aligned}
\mathcal{T}_1 &\equiv \{ s_{0;\alpha}^\pm (\lambda,t),  \quad \alpha\in \delta_0^+, \; \lambda \in l_0\} \cup  \{ s_{1;\alpha}^\pm (\lambda,t), \quad \alpha\in \delta_1^+, \; \lambda \in l_1\}, \\
\mathcal{T}_2 &\equiv \{ \tau_{0;\alpha}^\pm (\lambda,t),  \quad \alpha\in \delta_0^+, \; \lambda \in l_0\} \cup \{\tau_{1;\alpha}^\pm (\lambda,t),\quad \alpha\in \delta_1^+, \; \lambda \in l_1\}.
\end{aligned}\end{equation}

\begin{theorem}\label{thm:A51}
Let us assume that the potential of the Lax operator (\ref{eq:A51}) $Q(x,t)$ is a Schwartz-type function of $x$  and is such that the corresponding RHP is regular. Then each of the minimal sets $\mathcal{T}_i$, $i=1,2$ determines uniquely:

 i) all  sewing functions $G_\nu (x,t,\lambda)$ for $\nu =0,1, \dots , 11$;

 ii) all scattering matrices $T_\nu$, $\nu =0,1, \dots , 11$;

  iii) $\mathcal{T}_1 \simeq \mathcal{T}_2 $;

  iv) the potential $Q(x,t)$.
\end{theorem}

\begin{proof}[Idea of the proof]
The fact that the solution of the RHP is regular means that the corresponding Lax operator $L$ has no discrete eigenvalues. In other words the functions $D_\nu^\pm (\lambda)$ have neither zeroes nor poles in their regions of analyticity.

\begin{itemize}

   \item[i)] Let us now demonstrate that the sets $ \mathcal{T}_k$, $k=0, 1$ allow us to construct all $S_\nu^\pm (\lambda,t)$ and $T_\nu^\pm (\lambda,t)$. Obviously:
\begin{equation}\label{eq:S0S1}\begin{aligned}
S_0^\pm =&  \exp \left( s_{0;14}^\pm E_{\pm (e_1-e_4)} + s_{0;23}^\pm E_{\pm (e_2-e_3)} + s_{0;56}^\pm E_{\mp (e_5-e_6)} \right), \\
T_0^\pm =&  \exp \left( \tau_{0;14}^\pm E_{\pm (e_1-e_4)} + \tau_{0;23}^\pm E_{\pm (e_2-e_3)} + \tau_{0;56}^\pm E_{\mp (e_5-e_6)} \right), \\
S_1^\pm =&  \exp \left( s_{1;13}^\pm E_{\pm (e_1-e_3)} + s_{1;46}^\pm E_{\pm (e_4-e_6)} \right), \\ T_1^\pm =&  \exp \left( \tau_{1;13}^\pm E_{\pm (e_1-e_3)} + \tau_{1;46}^\pm E_{\pm (e_4-e_6)} \right).
\end{aligned}\end{equation}

Note that the reduction condition (\ref{eq:RC}) on the FAS reflects also on their asymptotics for $x\to \pm \infty$ as follows:
\begin{equation}\label{eq:Snu0}\begin{aligned}
C^\nu(S_{0}^\pm(x,t,\lambda)) =&  S_{2\nu }^\pm (x,t,\lambda \omega^{\nu}), &\quad
C^\nu(S_{1}^\pm(x,t,\lambda)) =&  S_{2\nu +1}^\pm (x,t,\lambda \omega^{\nu}), \\
C^\nu(T_{0}^\pm(x,t,\lambda)) =&  T_{2\nu }^\pm (x,t,\lambda \omega^{\nu}), &\quad
C^\nu(T_{1}^\pm(x,t,\lambda)) =&  T_{2\nu +1}^\pm (x,t,\lambda \omega^{\nu}),
\end{aligned}\end{equation}
for $\nu = 0,1 \dots, 11$. Thus we have recovered all $S_\nu^\pm (\lambda,t)$ and $T_\nu^\pm (\lambda,t)$.

   \item[ii)] It remains to recover $D_\nu^+ (\lambda)$ and $D_\nu^- (\lambda)$ (or $d_{\nu,\alpha}^\pm  (\lambda)$)  using the fact that they are analytic functions of $\lambda$ in the sector $\Omega_\nu$ and $\Omega_{\nu-1}$ respectively. In addition from eq. (\ref{eq:Tnu}) there follows, that:
\begin{equation}\label{eq:dpmT1}\begin{aligned}
d^+_{\nu; \alpha} - d^-_{\nu;\alpha} =&  \ln (1- s^+_{\nu,\alpha}  s^-_{\nu,-\alpha}), &\qquad \lambda \in  & l_\nu, &\quad \alpha \in & \delta_\nu^+ \\
d^+_{\nu; \alpha} - d^-_{\nu;\alpha} =&  \ln (1- \tau^+_{\nu,\alpha}  \tau^-_{\nu,\alpha}), &\qquad \lambda \in & l_\nu, &\quad \alpha \in  & \delta_\nu^+.
\end{aligned}\end{equation}
for $\nu = 0,1 \dots, 11$, which follow from eqs. (\ref{eq:Tnu}). In particular for $k=0, 1$:
\begin{equation}\label{eq:dpmT2}\begin{aligned}
d^+_{0; \alpha} - d^-_{0;\alpha} =&  \ln (1- s^+_{0,\alpha}  s^-_{0,-\alpha}), &\quad \lambda &\in l_0, &\quad \alpha &\in \{ e_1-e_4, e_2-e_3, -(e_5-e_6) \}; \\
d^+_{1; \alpha} - d^-_{1;\alpha} =&  \ln (1- s^+_{1,\alpha}  s^-_{1,\alpha}), &\quad \lambda &\in l_1, &\quad \alpha &\in \{  e_1 -e_3, e_4-e_6\},
\end{aligned}\end{equation}
and similar expressions in terms of $\tau^+_{k,\alpha}$ and $\tau^-_{k,-\alpha}$, $k=0, 1$.

   \item[iii)] Comparing the asymptotics (\ref{eq:limchi}) of the FAS for $x\to\pm \infty$ we easily find that the sewing functions $G_{\nu,0} $ in (\ref{eq:rhp2a}) are given by:
\begin{equation}\label{eq:Gnu0}\begin{aligned}
G_{k,0}(\lambda ,t) = \hat{S}_k^-(\lambda, t) S_k^+(\lambda,t) = \hat{D}_k^-(\lambda)  \hat{T}_k^+(\lambda, t) T_k^-(\lambda,t) D_k^+(\lambda ), \quad \lambda \in l_k, \; k=0,1.
\end{aligned}\end{equation}
Thus we know the left hand side of the relation:
\begin{equation}\label{eq:4.7}\begin{aligned}
D_k^-(\lambda ) G_{k,0} (\lambda,t) \hat{D}_k^+(\lambda,t) =  \hat{T}_k^+(\lambda, t) T_k^-(\lambda,t), \qquad k = 0, 1,
\end{aligned}\end{equation}
and the construction of $T_k^\pm (\lambda,t)$ reduces to decomposing the left hand side of (\ref{eq:4.7}) into Gauss factors, which has unique solution. This means that knowing $\mathcal{T}_1$ we can recover  $\mathcal{T}_2$. Quite analogously one can prove that knowing  $\mathcal{T}_2$ we can uniquely recover  $\mathcal{T}_1$.

   \item[iv)] The RHP has unique regular solution. Suppose we have constructed the solution $\xi_\nu (x,t, \lambda)$ in the sector $\Omega_\nu$. Then we recover the potential from the well known relation:
   \begin{equation}\label{eq:Qxt}\begin{aligned}
    Q(x,t) = \lim_{\lambda \to \infty} \lambda \left (J - \xi_\nu J \xi_\nu^{-1} (x,t,\lambda) \right).
   \end{aligned}\end{equation}
   This result does not depend on $\nu$ due to the reduction condition (\ref{eq:Zh0}) and to the fact that $C(Q(x,t)) = Q(x,t)$.

 \end{itemize}

\end{proof}

\subsection{The $A_5^{(2)}$ case}

We introduce two minimal sets of scattering data for the $A_5^{(2)}$ Kac-Moody algebra as follows, see Table \ref{tab:5}:
\begin{equation}\label{eq:T120}\begin{aligned}
\mathcal{T}_1 &\equiv \{ s_{0;\alpha}^\pm (\lambda,t),  \quad \alpha\in \delta_0^+, \; \lambda \in l_0\} \cup  \{ s_{1;\alpha}^\pm (\lambda,t), \quad \alpha\in \delta_1^+, \; \lambda \in l_1\}, \\
\mathcal{T}_2 &\equiv \{ \tau_{0;\alpha}^\pm (\lambda,t),  \quad \alpha\in \delta_0^+, \; \lambda \in l_0\} \cup \{\tau_{1;\alpha}^\pm (\lambda,t),\quad \alpha\in \delta_1^+, \; \lambda \in l_1\}.
\end{aligned}\end{equation}

\begin{theorem}\label{thm:A52}
Let us assume that the potential of the Lax operator (\ref{eq:A51}) $Q(x,t)$ is a Schwartz-type function of $x$  and is such that the corresponding RHP is regular. Then each of the minimal sets $\mathcal{T}_i$, $i=1,2$ determines uniquely:

i) all  sewing functions $G_\nu (x,t,\lambda)$ for $\nu =0,1, \dots , 19$;

ii) all scattering matrices $T_\nu$, $\nu =0,1, \dots , 19$;

iii) $\mathcal{T}_1 \simeq \mathcal{T}_2 $;

iv) the potential $Q(x,t)$.
\end{theorem}

\begin{proof}[Idea of the proof]
The fact that the solution of the RHP is regular means that the corresponding Lax operator $L$ has no discrete eigenvalues. In other words the functions $D_\nu^\pm (\lambda)$ have neither zeroes nor poles in their regions of analyticity.

\begin{itemize}

   \item[i)] Let us now demonstrate that the sets $ \mathcal{T}_k$, $k=0, 1$ allow us to construct all $S_\nu^\pm (\lambda,t)$ and $T_\nu^\pm (\lambda,t)$. Obviously:
\begin{equation}\label{eq:S0S1b}\begin{aligned}
S_0^\pm =&  \exp \left( s_{0;14}^\pm E_{\pm (e_1-e_4)}  \right), \\
T_0^\pm =&  \exp \left( \tau_{0;14}^\pm E_{\pm (e_1-e_4)}  \right), \\
S_1^\pm =&  \exp \left( s_{1;15}^\pm E_{\pm (e_1-e_5)} + s_{1;24}^\pm E_{\pm (e_2-e_4)} \right), \\ T_1^\pm =&  \exp \left( \tau_{1;15}^\pm E_{\pm (e_1-e_5)} + \tau_{1;24}^\pm E_{\pm (e_2-e_4)} \right).
\end{aligned}\end{equation}

Note that the reduction condition (\ref{eq:RC}) on the FAS reflects also on their asymptotics for $x\to \pm \infty$ as follows:
\begin{equation}\label{eq:Snu0b}\begin{aligned}
C^\nu(S_{0}^\pm(x,t,\lambda)) =&  S_{2\nu }^\pm (x,t,\lambda \omega^{\nu}), &\quad
C^\nu(S_{1}^\pm(x,t,\lambda)) =&  S_{2\nu +1}^\pm (x,t,\lambda \omega^{\nu}), \\
C^\nu(T_{0}^\pm(x,t,\lambda)) =&  T_{2\nu }^\pm (x,t,\lambda \omega^{\nu}), &\quad
C^\nu(T_{1}^\pm(x,t,\lambda)) =&  T_{2\nu +1}^\pm (x,t,\lambda \omega^{\nu}),
\end{aligned}\end{equation}
for $\nu = 0,1 \dots, 19$. Thus we have recovered all $S_\nu^\pm (\lambda,t)$ and $T_\nu^\pm (\lambda,t)$.

   \item[ii)] It remains to recover $D_\nu^+ (\lambda)$ and $D_\nu^- (\lambda)$ (or $d_{\nu,\alpha}^\pm  (\lambda)$)  using the fact that they are analytic functions of $\lambda$ in the sector $\Omega_\nu$ and $\Omega_{\nu-1}$ respectively. In addition from eq. (\ref{eq:Tnu}) there follows, that (see Table \ref{tab:5}):
\begin{equation}\label{eq:dpmT1b}\begin{aligned}
d^+_{\nu; \alpha} - d^-_{\nu;\alpha} =&  \ln (1- s^+_{\nu,\alpha}  s^-_{\nu,-\alpha}), &\qquad \lambda \in & l_\nu, &\quad \alpha \in & \delta_\nu^+ \\
d^+_{\nu; \alpha} - d^-_{\nu;\alpha} =&  \ln (1- \tau^+_{\nu,\alpha}  \tau^-_{\nu,\alpha}), &\qquad \lambda \in & l_\nu, &\quad \alpha \in & \delta_\nu^+.
\end{aligned}\end{equation}
for $\nu = 0,1 \dots, 19$, which follow from eqs. (\ref{eq:Tnu}). In particular for $k=0, 1$:
\begin{equation}\label{eq:dpmT2b}\begin{aligned}
d^+_{0; \alpha} - d^-_{0;\alpha} =&  \ln (1- s^+_{0,\alpha}  s^-_{0,-\alpha}), &\quad \lambda &\in l_0, &\quad \alpha &\in \{e_1-e_4  \}; \\
d^+_{1; \alpha} - d^-_{1;\alpha} =&  \ln (1- s^+_{1,\alpha}  s^-_{1,\alpha}), &\quad \lambda &\in l_1, &\quad \alpha &\in \{  e_1 -e_5, e_2-e_4\},
\end{aligned}\end{equation}
and similar expressions in terms of $\tau^+_{k,\alpha}$ and $\tau^-_{k,-\alpha}$, $k=0, 1$.

   \item[iii)] Comparing the asymptotics (\ref{eq:limchi}) of the FAS for $x\to\pm \infty$ we easily find that the sewing functions $G_{\nu,0} $ in (\ref{eq:rhp2a}) are given by:
\begin{equation}\label{eq:Gnu0b}\begin{aligned}
G_{k,0}(\lambda ,t) = \hat{S}_k^-(\lambda, t) S_k^+(\lambda,t) = \hat{D}_k^-(\lambda)  \hat{T}_k^+(\lambda, t) T_k^-(\lambda,t) D_k^+(\lambda ), \quad \lambda \in l_k, \; k=0,1.
\end{aligned}\end{equation}
Thus we know the left hand side of the relation:
\begin{equation}\label{eq:4.7b}\begin{aligned}
D_k^-(\lambda ) G_{k,0} (\lambda,t) \hat{D}_k^+(\lambda,t) =  \hat{T}_k^+(\lambda, t) T_k^-(\lambda,t), \qquad k = 0, 1,
\end{aligned}\end{equation}
and the construction of $T_k^\pm (\lambda,t)$ reduces to decomposing the left hand side of (\ref{eq:4.7b}) into Gauss factors, which has unique solution. This means that knowing $\mathcal{T}_1$ we can recover  $\mathcal{T}_2$. Quite analogously one can prove that knowing  $\mathcal{T}_2$ we can uniquely recover  $\mathcal{T}_1$.

   \item[iv)] The RHP has unique regular solution. Suppose we have constructed the solution $\xi_\nu (x,t, \lambda)$ in the sector $\Omega_\nu$. Then we recover the potential from the well known relation:
   \begin{equation}\label{eq:Qxt2}\begin{aligned}
    Q(x,t) = \lim_{\lambda \to \infty} \lambda \left (J - \xi_\nu J \xi_\nu^{-1} (x,t,\lambda) \right).
   \end{aligned}\end{equation}
   This result does not depend on $\nu$ due to the reduction condition (\ref{eq:Zh0}) and to the fact that $C(Q(x,t)) = Q(x,t)$.

 \end{itemize}

\end{proof}

\section{Discussion and conclusions}

We specified in \cite{GMSV} the choice od the corresponding Kac-Moody algebras and formulated the specific Lax operators and the corresponding direct and scattering problems. In each of the cases one needs to take into account specific peculiarities. For example in the case of $A_5^{(2)}$ after taking the average on the Coxeter automorphism the elements $B[2k-1,4]$ belong to the center of the algebra instead to its Cartan subalgebra.

The constructions that we outlined allow one to apply the dressing Zakharov-Shabat method and derive the soliton solutions of the corresponding mKdV and 2-dimensional Toda field theories. One may expect additional difficulties in this, due to the fact that the Coxeter symmetries require that even the simplest dressing factors must contain at leat $2h$ simple poles (that is 12 and 20 poles) whose residues $P_k$ must be related by the Coxeter automorphism. Therefore it will be important that deriving the projectors we must strictly stick to the construction of the FAS in each of the sectors of analyticity.

The main ideas in this and many previous publications of the author (see e.g. \cite{Ger, VG-13, Ge}) are based on the notion of fundamental analytic solution introduced by A. B. Shabat \cite{Sha*75, Sha}.

Another important trend started by A. B. Shabat and his collaborators concerns the classification of the integrable NLEE, see  \cite{MSY, Sha3, ShaYam, ASYa, SMS, MWN1, MWN2} and the numerous references therein. The idea is based on the theorem that if a given NLEE possesses master symmetry then it has an infinite number of integrals of motion and therefore must be integrable.

The final remark here concerns the fact that  the one to one correspondence between the minimal sets of scattering data and the potential $Q(x,t)$ follows also from the expansions over the squared solutions of $L$, see \cite{Ger, Ge, VG-13, GYa, SIAM*14}. These ideas will be published elsewhere.

\section*{Acknowledgements}
I am grateful to Dr Alexander Stefanov and Dr Stanislav Varbev for useful discussions.

\appendix
\section{Basis and grading of $A_{5}^{(1)}$}

The rank of the algebra $A_{5}^{(1)}\simeq sl(6)$ is $5$, the Coxeter number is $h=6$ and its exponents are $1, 2, 3, 4, 5$.
 The root system and the set of simple roots $\alpha_j$ of $A_{5}^{(1)} \simeq sl(6)$ is
 \begin{equation}\label{eq:A51RS}\begin{aligned}
 \Delta \equiv &\Delta^+ \cup \Delta^-, &\qquad \Delta^\pm \equiv & \{ \pm (e_j - e_k), \quad 1 \leq j < k \leq 6 \} ,  \\
  \alpha_j =&  e_j - e_{j+1}, &\qquad j =& 1,\dots, 5.
 \end{aligned}\end{equation}

 The Cartan-Weyl basis of $A_{5}^{(1)}$ in the typical representation is as follows:
 \begin{equation}\label{eq:A51CW}\begin{aligned}
 H_{e_j - e_k} =&  E_{jj} - E_{kk}, &\qquad E_{e_j - e_k} &= E_{jk}, &\qquad E_{-\alpha}= &E_\alpha^T, \\
 [H_\alpha, E_\beta] =&  (\alpha, \beta) E_\beta, &\quad  & & [E_\alpha, E_\beta] =&  N_{\alpha,\beta} E_{\alpha+\beta}.
 \end{aligned}\end{equation}
 The numbers $N_{\alpha,\beta} = - N_{\beta, \alpha}$; they are non-vanishing if and only if the sum of roots $\alpha +\beta \in \Delta$.

The Dynkin diagram  of  $A_5$ algebra, and the extended Dynkin diagrams of $A_{5}^{(1)}$ and $A_{5}^{(2)}$  are shown in Figure \ref{fig:1}.


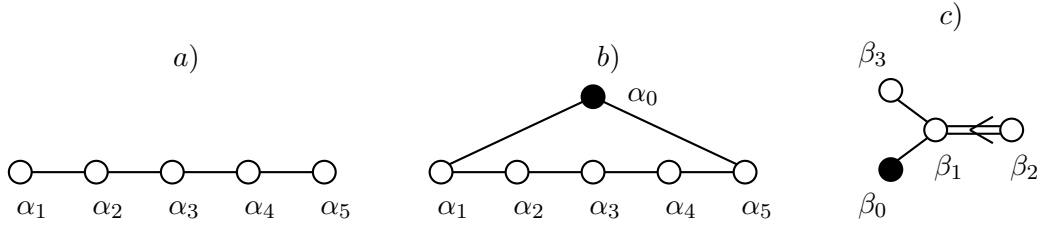
\begin{figure}
\begin{center}
\begin{minipage}[b]{0.3\textwidth}
  \begin{tikzpicture}[scale=0.5]
    \draw (7,3) node[anchor=east]  {$a)$};
   \draw (3,-1) node[anchor=east]  {$\alpha_1$};
   \draw (5,-1) node[anchor=east]  {$\alpha_2$};
   \draw (7,-1) node[anchor=east]  {$\alpha_3$};
   \draw (9,-1) node[anchor=east]  {$\alpha_4$};
   \draw (11,-1) node[anchor=east]  {$\alpha_5$};
    \foreach \x in {1,...,5}
    \draw[xshift=\x cm,thick] (\x cm,0) circle (.3cm);
     \foreach \y in {1.15,...,4.15}
    \draw[xshift=\y cm,thick] (\y cm,0) -- +(1.4 cm,0);
	\foreach \z in {1,...,5}
	\coordinate (\z cm,-2cm);
  \end{tikzpicture}
\end{minipage}
\begin{minipage}[b]{0.03\textwidth}
\hfill
\end{minipage}
\begin{minipage}[b]{0.3\textwidth}
  \begin{tikzpicture}[scale=0.5]
    \draw (5,3) node[anchor=east]  {$b)$};
     \draw (6,2) node[anchor=east]  {$\alpha_0$};
   \draw (1,-1) node[anchor=east]  {$\alpha_1$};
   \draw (3,-1) node[anchor=east]  {$\alpha_2$};
   \draw (5,-1) node[anchor=east]  {$\alpha_3$};
   \draw (7,-1) node[anchor=east]  {$\alpha_4$};
   \draw (9,-1) node[anchor=east]  {$\alpha_5$};
    \foreach \x in {0,...,4}
    \draw[thick,xshift=\x cm] (\x cm,0) circle (3 mm);
    \foreach \y in {0,...,3}
    \draw[thick,xshift=\y cm] (\y cm,0) ++(.3 cm, 0) -- +(14 mm,0);
    \draw[thick,fill=black] (4 cm,2 cm) circle (3 mm);
 \draw[thick] (2 mm, 2 mm) -- +(3.6 cm, 1.7 cm);
 \draw[thick] (4.2 cm, 1.9 cm) -- +(3.6 cm, -1.7 cm);
  \end{tikzpicture}
\end{minipage}
\begin{minipage}[b]{0.03\textwidth}
\hfill
\end{minipage}
\begin{minipage}[b]{0.3\textwidth}
  \begin{tikzpicture}[scale=0.5]
    \draw (9,3) node[anchor=east]  {$c)$};
     \draw (7,-2) node[anchor=east]  {$\beta_0$};
   \draw (9,-1) node[anchor=east]  {$\beta_1$};
   \draw (11,-1) node[anchor=east]  {$\beta_2$};
   \draw (7,2) node[anchor=east]  {$\beta_3$};
\draw[xshift=4 cm,thick] (4 cm,0) circle (.3cm);
\draw[xshift=5 cm,thick] (5 cm, 0) circle (.3 cm);
  \draw[xshift=5 cm,thick] (30: 21 mm) circle (.3cm);
    \draw[xshift=5 cm,thick,fill=black] (-30: 21 mm) circle (.3cm);
  \draw[thick] (7 cm, 0.8 cm) -- +(0.8 cm, -0.6 cm);
  \draw[thick] (7 cm, -0.8 cm) -- +(0.8 cm, 0.6 cm);
   \draw[thick] (8.3 cm, .1 cm) -- +(1.4 cm,0);
    \draw[thick] (8.3 cm, -.1 cm) -- +(1.4 cm,0);
   \draw[xshift=8.9 cm,thick] (15: 0 mm) -- (30: 7 mm);
    \draw[xshift=8.9 cm,thick] (-15: 0 mm) -- (-30: 7 mm);
  \end{tikzpicture}
\end{minipage}
 \caption{Dynkin diagrams (DD) of $A_5$ and related Kac-Moody algebras: a) DD of $A_5 \simeq sl(6)$; b) extended DD of $A_5$; c) DD of $A_5^{(2)}$. }\label{fig:1}
\end{center}
\end{figure}

Let us now briefly outline  how to define Kac-Moody algebra starting from a simple Lie algebra $\mathfrak{g}$ which in our case is chosen to be $A_5 \simeq sl(6)$. First we  use a Coxeter automorphism to introduce a grading in the Lie algebra $A_{5}$:
\begin{equation*}
\begin{aligned}
\mathfrak{g} =&  \mathop{\oplus}\limits_{k=0}^{5} \mathfrak{g}^{(k)}, &\qquad
\tilde{\mathfrak{g}} =&  \mathop{\oplus}\limits_{s=0}^{5} \tilde{\mathfrak{g}}_{s},
\end{aligned}
\end{equation*}
where the linear subspaces are such that
\begin{equation*}
\begin{aligned}
C_1 X C_1^{-1} =&  \omega_1 ^{-k} X, \quad \forall  X&\in \mathfrak{g}^{(k)},
 \qquad  \tilde{ C}_1 Y \tilde{ C}_1^{-1} =&  \omega^{-s} Y, \quad \forall Y\in \tilde{ \mathfrak{g}}_{s},
\end{aligned}
\end{equation*}
where $\omega_1 = e^{2\pi i/6}$. Each of the gradings satisfies
\begin{equation}\label{eq:grad}
[ \mathfrak{g}^{(k)}, \mathfrak{g}^{(m)}] \in \mathfrak{g}^{(k+m)}, \qquad
[ \tilde{ \mathfrak{g}}_{s}, \tilde{ \mathfrak{g}}_{p}] \in \tilde{ \mathfrak{g}}_{s+p},
\end{equation}
where $(k+m)$ and $(s+p)$ must be understood modulo $6$. The indices for $A_{5}^{(1)}$ are everywhere taken modulo $6$. Using this grading we can now construct polynomials in $\lambda$ and $\lambda^{-1}$ such that:
\begin{equation}\label{eq:KM}\begin{aligned}
 X(\lambda) = \sum_{s=-\infty}^{N} \lambda^s X_s, \qquad X_s \in \mathfrak{g}^{(s)}.
\end{aligned}\end{equation}
which will be the elements of the Kac-Moody algebra \cite{Kac, Cart}.
Here the upper index of the subspace $s$ is evaluated modulo 6. Obviously the commutator of two such polynomials in $\lambda$ and $\lambda^{-1}$ due to the properties of the grading (\ref{eq:grad}) will again be of the form (\ref{eq:KM}). Of course the rigorous definition of Kac-Moody  algebra requires additional structures, which we will not mention now.

For the case of $A_{5}$ algebra are possible two different types of Coxeter's automorphisms which will give two Kac-Moody algebras  $A_{5}^{(1)}$ with height $1$ and  $A_{5}^{(2)}$ with height $2$.

There are two standard choices $C_{1}$ and $\tilde{C}_{1}$ for the Coxeter  automorphism for the algebra $\mathfrak{g}\simeq A_5$. This is $\mathbb{Z}_{6}$ automorphism. With this automorphism we effectively work with Kac-Moody algebra $A_{5}^{(1)}$. Indeed, each of these choices satisfies $C_1^{6}=\openone$,  $\tilde{ C}_1^{6}=\openone$ and each of these automorphisms induces a grading in $\mathfrak{g}$

In what follows it is specified the choice of the automorphisms by
\begin{equation}
C_1 =  \left(\begin{array}{cccccc}
0 & 1 & 0 & 0 & 0 & 0 \\ 0 & 0 & 1 & 0 & 0 & 0 \\ 0 & 0 & 0 & 1 & 0 & 0 \\ 0 & 0 & 0 & 0 & 1 & 0 \\ 0 & 0 & 0 & 0 & 0 & 1 \\ -1 & 0 & 0 & 0 & 0 & 0  \end{array}\right),
\qquad \tilde{C}_1 = \left(\begin{array}{cccccc} 1 & 0 & 0 & 0 & 0 & 0 \\ 0 & \omega_1  & 0 & 0 & 0 & 0 \\ 0 & 0 &  \omega_1 ^2 & 0 & 0 & 0 \\ 0 & 0 & 0 & \omega_1 ^3 & 0 & 0 \\ 0 & 0 & 0 & 0 & \omega_1 ^4 & 0 \\ 0 & 0 & 0 & 0 & 0 & \omega_1 ^5  \end{array}\right).
\end{equation}
Obviously $C^6 = \tilde{C}^6 = \openone$.
Below we will use also the notations $C_1= J_{1}^{(0)}$ and $\tilde{C}_1 =J_{0}^{(1)}$ along with the more general ones $J_{s}^{(k)}$ which provide a convenient basis in $A^{(1)}_{5}$ which satisfies the above gradings, see  \cite{Bourb, Helg, Cart, Kac}:
\begin{equation*}
J_s^{(k)} = \sum_{j=1}^{6} \epsilon_{j,j+s} \omega_1  ^{-k(j-1)} E_{j,j+s}., \qquad  \epsilon_{j,j+s} = \begin{cases} 1 & \quad \mbox{if} j+s \leq 6 \\ -1 & \mbox {if} \quad j+k>6. \end{cases}.
\end{equation*}
Here the ($6\times6$) matrices $E_{km}$ are defined by $(E_{km})_{sp} = \delta_{ks}\delta_{mp}$.

The elements  of this basis satisfy the commutation relations
\begin{equation*}
\left[ J_s^{(k)} , J_l^{(m)} \right] = \left(\omega_1^{-ms} - \omega_1^{-kl} \right) J_{s+l}^{(k+m)}.
\end{equation*}
Besides it is easy to check that
\begin{equation*}
 C^{-1}_{1} J_s^{(k)} C_{1} = \omega_1 ^{-k} J_s^{(k)}, \qquad \tilde{ C}^{-1}_{1} J_s^{(k)}
 \tilde{ C}_{1} = \omega_1 ^{-s} J_s^{(k)}
\end{equation*}
and
\begin{equation*}
J_{s}^{(k)}J_{p}^{(m)} = \omega_1 ^{-sm} J_{s+p}^{(k+m)}, \qquad (J_{s}^{(k)})^{-1}=(J_{s}^{(k)})^{\dagger}.
\end{equation*}

Using this, the bases in each of the linear subspaces can be specified as follows
\begin{equation*}
\mathfrak{g}^{(k)} \equiv \mbox{l.c.\;} \{ J_s^{(k)}, \quad s=1, \dots, 6\}, \qquad
\tilde{ \mathfrak{g}}_{s} \equiv \mbox{l.c.\;} \{ J_s^{(k)}, \quad k=1, \dots, 6\}.
\end{equation*}

The basis that we constructed for $A_{5}^{(1)}$ is
\begin{equation*}
\begin{aligned}
\mathfrak{g}^{(0)} &: \lc \{J_1^{(0)} , J_2^{(0)}, J_3^{(0)}, J_4^{(0)}, J_5^{(0)}\}, &\quad
\mathfrak{g}^{(1)} &: \lc \{ J_1^{(1)} , J_2^{(1)}, J_3^{(1)}, J_4^{(1)}, J_5^{(1)},  J_6^{(1)}\}, \\
\mathfrak{g}^{(2)} &: \lc \{ J_1^{(2)} , J_2^{(2)}, J_3^{(2)}, J_4^{(2)}, J_5^{(2)},  J_6^{(2)}\}, &\quad
\mathfrak{g}^{(3)} &: \lc \{J_1^{(3)} , J_2^{(3)}, J_3^{(1)}, J_4^{(3)}, J_5^{(3)},  J_6^{(3)}\}, \\
\mathfrak{g}^{(4)} &: \lc \{ J_1^{(4)} , J_2^{(4)}, J_3^{(4)}, J_4^{(4)}, J_5^{(4)},  J_6^{(4)}\}, &\quad
\mathfrak{g}^{(5)} &: \lc \{ J_1^{(5)} , J_2^{(5)}, J_3^{(5)}, J_4^{(5)}, J_5^{(5)},  J_6^{(5)}\}.
\end{aligned}
\end{equation*}

\section{Basis and grading of $A_{5}^{(2)}$}

Let us now briefly outline the gradings for $A_{5}^{(2)}$. Now we will use as Coxeter automorphism $C_2 =C_1\circ V$ which is a composition of $C_1$ with the external automorphism $V$ of $A_5$.
$V$  is  generated by the symmetry of its Dynkin diagram. In the five-dimensional space of roots $V$ maps $V: e_k \to -e_{7-k}$, $k=1,\dots, 6$. On any of the root vectors  $X$, $V$ acts as
\begin{equation}\label{eq:V1b}\begin{aligned}
V(X) = - S_2X^T S_2^{-1}, \qquad  S_2 = E_{1,6} -E_{2,5} +E_{3,4} -E_{4,3} + E_{5,2} -E_{6,1} .
\end{aligned}\end{equation}
Note that $ S_2^{-1} = -S_2$. Obviously $V$ splits the Lie algebra $\mathfrak{g} \simeq A_5$ into two: $\mathfrak{g} = \mathfrak{g}_0 \cup \mathfrak{g}_1$ whose bases, corresponding to the positive roots  are given as follows:
\begin{equation}\label{eq:V2b}\begin{aligned}
&  \mathfrak{g}^{(0)} \colon &\quad & \{ \tilde{\mathcal{E}}_{ij}^+, \quad  \tilde{\mathcal{E}}_{i\bar{j}}^+, \quad \tilde{\mathcal{E}}_{j\bar{j}}^+ ,\quad 1\leq i <j \leq 3 \}, \\
& \mathfrak{g}^{(1)} \colon &\quad & \{ \tilde{\mathcal{E}}_{ij}^-, \quad  \tilde{\mathcal{E}}_{i\bar{j}}^-, \quad  1\leq i <j \leq 3\}, \\
& &  & \tilde{\mathcal{E}}_{ij}^\pm = E_{ij} \mp  (-1)^{i-j} E_{\bar{j}, \bar{i}}, \quad
\tilde{\mathcal{E}}_{i\bar{j}}^\pm = E_{i\bar{j}} \pm  (-1)^{i-j} E_{j, \bar{i}}, \quad
\tilde{\mathcal{E}}_{j\bar{j}}^+ = E_{j\bar{j}}
\end{aligned}\end{equation}
Here we can identify the root vectors:
\begin{equation}\label{eq:Eeij}\begin{aligned}
 E_{e_i - e_j}^\pm  =&   \tilde{\mathcal{E}}_{ij}^\pm, &\quad  E_{e_i + e_j}^\pm  =&   \tilde{\mathcal{E}}_{i\bar{j}}^\pm, \quad E_{2e_j}^+ =&  \tilde{\mathcal{E}}_{j\bar{j}}^+ .
\end{aligned}\end{equation}
Obviously $ E_{e_i - e_j}^+$, $ E_{e_i + e_j}^+$ and $E_{2e_j}^+$ are the generators of $sp(6)$ corresponding to its positive roots; $ E_{e_i - e_j}^-$ and $ E_{e_i + e_j}^-$ provide the positive roots of $\mathfrak{g}_1$. It is easy to check that they satisfy standard commutation relations, taking into account the $\mathbb{Z}_2$-grading such as:
\begin{equation}\label{eq:Ealp}\begin{aligned}
{} [  E_\alpha^\pm, E_{-\alpha}^\pm] =&  H_\alpha, &\quad  [H, E_\alpha^\pm ] = & \alpha(H)  E_\alpha^\pm , \\
[E_\alpha^-, E_\beta^- ] =&  n_{\alpha,\beta}^- E_{\alpha+\beta}^+, &\quad
[E_\alpha^-, E_\beta^+ ] =&  n_{\alpha,\beta}^+ E_{\alpha+\beta}^-, &\quad
\end{aligned}\end{equation}
etc. Let us now take into account the Coxeter automorphism which is given by
$C_2(X) = C_1 V(X)C_1^{-1} = -C_1 S_2 X^T S_2^{-1} C_1^{-1}$,
One can check that $C_2^{10} =\openone$, so the Coxeter number is $h_2=10$.
$C_2$ splits the  roots of $A_5$ into three orbits  each containing 10 roots.
The grading condition is
$\big[ \mathfrak{g}^{(k)}, \mathfrak{g}^{(l)} \big] \subset \mathfrak{g}^{(k+l)}$, 
$k, l = 1,\dots, 10$,
where $k+l$ is taken modulo $10$. We assume that the orbits start from the root vectors $E_{12}$, $E_{34}$ and $E_{13}$. We consider also the action of $C_2$ also on the Cartan generators.    The basis for each of the subspaces $\mathfrak{g}^{(k)}$ is obtained by taking the weighted average over the action of $C_2$:
\begin{equation}\label{eq:Eijk}\begin{aligned}
\mathcal{E}_{ij}^{(k)} = \sum_{s=0}^{9} \omega_2^{-ks} C_2^s(E_{ij}), \qquad \mathcal{H}_{1}^{(k)} = \sum_{s=0}^{9} \omega_2^{-ks} C_2^s(E_{11}), \qquad \omega_2 = \exp(2\pi i/10).
\end{aligned}\end{equation}
It is easy to check that
$ C_2 (\mathcal{E}_{ij}^{(k)}) = \omega_2^k \mathcal{E}_{ij}^{(k)}$, $ C_2 (\mathcal{H}_{1}^{(k)}) = \omega_2^k \mathcal{H}_{1}^{(k)}$,
i.e. $\mathcal{E}_{ij}^{(k)})$ and $ \mathcal{H}_{1}^{(k)}$ belong to $\mathfrak{g}^{(k)}$. We will provide this basis explicitly:
\begin{equation}\label{eq:EEijk}\begin{aligned}
\mathcal{E}_{12}^{(k)} =&  \left(\begin{array}{cccccc} 0 & 1 & 0 & 0 & -\omega_2^{-3k} & 0 \\  -\omega_2^{-5k} & 0 & \omega_2^{-2k} & 0 & 0 & 0 \\  0 & -\omega_2^{-7k} & 0 & 0 & 0 & -\omega_2^{-4k} \\  0 & 0 & 0 & 0 & 0 & 0 \\  \omega_2^{-8k} & 0 & 0 & 0 & 0 & -\omega_2^{-k} \\  0 & 0 & \omega_2^{-9k} & 0 & \omega_2^{-6k} & 0  \end{array}\right), \\
\mathcal{E}_{34}^{(k)} =&  \left(\begin{array}{cccccc} 0 & 0 & 0 & \omega_2^{-6k} & 0 & 0 \\  0 & 0 & 0 & -\omega_2^{-8k} & 0 & 0 \\  0 & 0 & 0 & 1 & 0 & 0 \\  -\omega_2^{-k} & \omega_2^{-3k} & -\omega_2^{-5k} & 0 & \omega_2^{-9k} & -\omega_2^{-7k} \\  0 & 0 & 0 & -\omega_2^{-4k} & 0 & 0 \\  0 & 0 & 0 & \omega_2^{-2k} & 0 & 0  \end{array}\right), \\
\mathcal{E}_{13}^{(k)} =&  \left(\begin{array}{cccccc} 0 & 0 & 1 & 0 & 0 & -\omega_2^{-k} \\  0 & 0 & 0 & 0 & -\omega_2^{-3k} & -\omega_2^{-2k} \\  -\omega_2^{-5k} & 0 & 0 & 0 & -\omega_2^{-4k} & 0 \\  0 & 0 & 0 & 0 & 0 & 0 \\  0 & \omega_2^{-8k} & \omega_2^{-9k} & 0 & 0 & 0 \\  \omega_2^{-6k} & \omega_2^{-7k}& 0 & 0 & 0 & 0  \end{array}\right), \\
\mathcal{H}_{1}^{(k)} =&  c_{H,k} \diag ( \omega_2^{-5k},  \omega_2^{-7k}, \omega_2^{-9k}, 0,  \omega_2^{-3k}, \omega_2^{-k}  ).
\end{aligned}\end{equation}
where $c_{H,k} =\omega_2^{5k}-1$. Since $\omega_2^5 =-1$ it is easy to see that $c_{H,k} \neq 0$ for $k =1,3,5,7$ and 9. Thus the subspaces $\mathfrak{g}^{(p)}$ have nontrivial section with the Cartan subalgebra if and  only if $p$ is an exponent of $A_5^{(2)}$.
It is easy to check that $\tilde{\mathcal{E}}_{ij}^{+}$ provide a basis for the subalgebra $sp(6)$ of $A_{5}^{(2)}$.
Then the  basis in each of the subspaces $\mathfrak{g}^{(k)}$ is as follows
\begin{equation}\label{eq:gr2c}
\begin{aligned}
\tilde{\mathfrak{g}}^{(0)} =&  \lc \{ \tilde{\mathcal{E}}_{11}^{+}, \tilde{\mathcal{E}}_{22}^{+}, \tilde{\mathcal{E}}_{33}^{+}\}, &\quad
\tilde{\mathfrak{g}}^{(1)} =&  \lc \{ \tilde{\mathcal{E}}_{21}^{+}, \tilde{\mathcal{E}}_{32}^{+}, \tilde{\mathcal{E}}_{43}^{+}, \tilde{\mathcal{E}}_{15}^{-}\}, \\
\tilde{\mathfrak{g}}^{(2)} =&  \lc \{  \tilde{\mathcal{E}}_{31}^{+}, \tilde{\mathcal{E}}_{42}^{+}, \tilde{\mathcal{E}}_{14}^{-}\}, &\quad
\tilde{\mathfrak{g}}^{(3)} =&  \lc \{\tilde{\mathcal{E}}_{14}^{+}, \tilde{\mathcal{E}}_{25}^{+}, \tilde{\mathcal{E}}_{13}^{-}, \tilde{\mathcal{E}}_{24}^{-}\}, \\
\tilde{\mathfrak{g}}^{(4)} =&  \lc \{ \tilde{\mathcal{E}}_{51}^{+}, \tilde{\mathcal{E}}_{12}^{-}, \tilde{\mathcal{E}}_{23}^{-}\}, &\quad
\tilde{\mathfrak{g}}^{(5)} =&  \lc \{ \tilde{\mathcal{E}}_{16}^{+}, \tilde{\mathcal{E}}_{61}^{+}, \tilde{\mathcal{E}}_{33}^{-}-\tilde{\mathcal{E}}_{11}^{-}, \tilde{\mathcal{E}}^{-}_{33} -\tilde{\mathcal{E}}_{22}^{-}\},\\
\tilde{\mathfrak{g}}^{(6)} =&  \lc \{ \tilde{\mathcal{E}}_{15}^{+}, \tilde{\mathcal{E}}_{21}^{-}, \tilde{\mathcal{E}}_{32}^{-}\}, &\quad
\tilde{\mathfrak{g}}^{(7)} =&  \lc \{ \tilde{\mathcal{E}}_{41}^{+}, \tilde{\mathcal{E}}_{52}^{+}, \tilde{\mathcal{E}}_{31}^{-}, \tilde{\mathcal{E}}_{42}^{-}\},\\
\tilde{\mathfrak{g}}^{(8)} =&  \lc \{ \tilde{\mathcal{E}}_{13}^{+}, \tilde{\mathcal{E}}_{24}^{+}, \tilde{\mathcal{E}}_{41}^{-}\}, &\quad
\tilde{\mathfrak{g}}^{(9)} =&  \lc \{ \tilde{\mathcal{E}}_{12}^{+}, \tilde{\mathcal{E}}_{23}^{+}, \tilde{\mathcal{E}}_{34}^{+}, \tilde{\mathcal{E}}_{51}^{-}\}.
\end{aligned}
\end{equation}

As a result the rank of $A_{5}^{(2)}$ is $3$, $h=10$ and its exponents are $1, 3, 5, 7, 9$, see \cite{DS, Cart}.

An alternative grading of $A_5^{(2)}$ can be achieved by using a realization of the Coxeter automorphism as an element of the Cartan subgroup. More precisely, one can use the automorphism $\tilde{C}_2$ \cite{DS}:
\begin{equation}\label{eq:C2t}\begin{aligned}
 \tilde{C}_2 (X) = -S_2 X^T S_2^{-1}. \quad S_2 = \diag (1, -\omega_2 , \omega_2^2, -\omega_2^3, \omega_2^4, -\omega_2^5 ), \quad \omega_2 = \exp(2\pi i/10),
\end{aligned}\end{equation}
and where the transposition must be taken with respect to the second diagonal of the matrix. With choice for the Coxeter automorphism the set of admissible roots of $A_5^{(2)}$ acquire the form:
\begin{equation}\label{eq:Etij}\begin{aligned}
 \mathcal{E}_{\beta_0} =&  \frac{\zeta}{2}(E_{1,5}+E_{2,6}), &\quad  \mathcal{E}_{-\beta_0} =&  2(E_{5,1}+E_{6,2}) \zeta^{-1}, & \quad \mathcal{H}_{\beta_0} =&  \mathcal{H}_1 +\mathcal{H}_{2}, \\
 \mathcal{E}_{\beta_i} =&  \zeta (E_{i+1,i}+E_{7-i,6-i}), &\quad  \mathcal{E}_{-\beta_i} =&  (E_{i,i+1}+E_{6-i,7-i}) \zeta^{-1}, & \quad \mathcal{H}_{\beta_i} =&
 \mathcal{H}_{i+1} -\mathcal{H}_{i},\\
 \mathcal{E}_{\beta_3} =&  \zeta E_{4,3}, &\quad  \mathcal{E}_{-\beta_3} =&  E_{3,4} \zeta^{-1}, & \quad \mathcal{H}_{\beta_i} =&  -\mathcal{H}_{3} +\mathcal{H}_{4},
\end{aligned}\end{equation}
where $i=1, 2$, $(E_{km})_{ab} = \delta_{ka} \delta_{mb}$ and $ \mathcal{H}_1 = E_{i,i} - E_{7-i,7-i}$.

\end{document}